\documentclass{article}
\usepackage{verbatim}

\usepackage{times}
\usepackage[pdftex]{graphicx}
\usepackage{amssymb,amsfonts,amsmath,amsthm}
\usepackage{float}
\usepackage{xcolor}
\usepackage{enumitem}

\newtheorem{theorem}{Theorem}
\newtheorem{proposition}{Proposition}
\newtheorem{lemma}{Lemma}
\newtheorem{ass}{Assumption}
\newtheorem{definition}{Definition}
\newtheorem{remark}{Remark}
\newtheorem{cor}{Corollary}

\newcommand\EE {\mathbb E}
\newcommand\FF {\mathbb F}
\newcommand\NN {\mathbb N}
\newcommand\RR {\mathbb R}
\newcommand\PP {\mathbb P}
\newcommand\QQ {\mathbb Q}

\newcommand\ZZ {\mathbb Z}

\newcommand\cA {\mathcal A}

\newcommand\cL {\mathcal L}

\newcommand\cF {\mathcal F}

\def\bone{\mathbf{1}}

\def\D{\Delta}

\def\pa{\partial}
\def\e{\epsilon}
\def\a{\alpha}
\def\d{\delta}
\newcommand\1 {\mathbf 1}
\def\qed{\hskip6pt\vrule height6pt width5pt depth1pt}

\def\qed{\hskip 6pt\vrule height6pt width5pt depth1pt}

\newcommand{\ed}{\end{document}}
\newcommand{\be}{\begin{equation}}
\newcommand{\ee}{\end{equation}}
\newcommand{\bq}{\begin{eqnarray}}
\newcommand{\eq}{\end{eqnarray}}

\vspace{4in}

\begin{document}

\title{Utility-based pricing and hedging of contingent claims in Almgren-Chriss model with temporary price impact.}
\author{Ibrahim~Ekren and Sergey~Nadtochiy\footnote{We thank Yavor Stoev for his help at the initial stage of this project.
S. Nadtochiy received partial support from the NSF CAREER grant 1855309.
Data sharing is not applicable to this article as no new data were created or analyzed in this study.
Address the correspondence to: Department of Applied Mathematics, Illinois Institute of Technology, 10 W. 32nd St., Chicago, IL 60616 (snadtochiy@iit.edu).
}}
\date{First draft: October 3, 2019,\\
Current version: June 15, 2020.
}

\maketitle

\begin{abstract}
In this paper, we construct the utility-based optimal hedging strategy for a European-type option in the Almgren-Chriss model with temporary price impact. The main mathematical challenge of this work stems from the degeneracy of the second order terms and the quadratic growth of the first order terms in the associated HJB equation, which makes it difficult to establish sufficient regularity of the value function needed to construct the optimal strategy in a feedback form. By combining the analytic and probabilistic tools for describing the value function and the optimal strategy, we establish the feedback representation of the latter. We use this representation to derive an explicit asymptotic expansion of the utility indifference price of the option, which allows us to quantify the price impact in options' market via the price impact coefficient in the underlying market.
\end{abstract}

\section{Introduction}

This paper is concerned with the problem of hedging and pricing of contingent claims in a model with price impact. More specifically, we restrict our analysis to European-type claims and assume the Almgren-Chriss model (see \cite{AlmgrenChriss}) with linear temporary impact for the underlying asset. We also assume that the preferences of the agent (performing the hedging or pricing of the option) are given by an exponential utility. Then, the optimal hedging strategy is determined by maximizing the expected exponential utility of the terminal wealth generated by the dynamic trading in the underlying plus the payoff of the option. A natural notion of option price, in this setting, is the utility indifference price (see Definition \ref{def:indif.price}), which can be computed via the value function of the aforementioned maximization problem.  

The problem of hedging of contingent claims in the Almgren-Chriss model (and in its extensions with nonlinear price impact) has been studied before. Much of the existing literature is concerned with the problems or replication and super-replication of contingent claims: see, e.g., \cite{BankBaum, CvitanicMa, Yong, Loeper}, and the references therein. However, the optimal (super-)replication strategies are only constructed in the models with permanent impact -- i.e. without temporary one -- and the exact replication strategies typically do not exist in the presence of temporary impact. An optimal hedging strategy is constructed in \cite{AlmgrenHedging, BankVS, BV, GP1, GP2}, but for an agent maximizing a linear-quadratic objective. The latter objective suffers from several shortfalls: in particular, it penalizes the hedger for making profits and may produce static arbitrages in the options' prices. Our setting is close to the one of \cite{GuP}, which poses the hedging problem as the maximization of expected exponential utility. However, \cite{GuP} does not provide a complete well-posedness theory for the associated Hamilton-Jacobi-Bellman (HJB) equation (the validity of comparison principle is left open), and, more importantly, it does not provide a rigorous characterization of the optimal hedging strategy. The reason for the latter, as well as for the lack of characterization of the optimal super-replicating strategies in \cite{Loeper} and in other models including temporary impact, is that the associated HJB equation (or its stochastic analogue) is degenerate and has a quadratic nonlinearity in the gradient. This makes it difficult to establish the desired regularity of its solution, needed to construct the optimal strategy in feedback form (and the well-posedness of the associated forward-backward systems is not even clear). The main contribution of this paper is in providing an explicit and computationally tractable characterization of the optimal hedging strategy in feedback form. The latter is achieved by combining the analysis of the associated HJB equation, the direct properties of the stochastic optimization problem (in particular, its strong convexity), and the representation of the optimal control via a Backward Stochastic Differential equation (BSDE), in order to establish the so-called ``endogenous boundedness": i.e., the optimal control is bounded by a constant, even though no a priori constraints on its values are imposed in the optimization problem. The latter result is summarized in Theorem \ref{unifbound}, and it allows us to complete the description of optimal control in the feedback form and obtain Theorem \ref{theod0}.

Another contribution of the present paper is in the analysis of utility indifference price of an option. In particular, we provide a computationally tractable description of this price via the HJB equation for the value function and, more importantly, develop rigorously the asymptotic expansion of this price in the regime where the price impact in the underlying market is small (see Theorem \ref{theo:expansion}). To understand the value of this result, assume that the underlying market is sufficiently liquid, so that that price impact coefficient in this market, denoted $\eta$, can be measured. The option's market, on the other hand, is less liquid, and the trading occurs via a market maker, who buys from, or sells to, a client a certain number of shares of the option and hedges her position by trading in the underlying market. Then, the market maker plays the role of the aforementioned agent, and it is natural to assume that she will trade the option's shares at her utility indifference price (see the next paragraph for a justification of this assumption). Recall that the indifference price of the option depends on $\eta$, as the latter affects the hedging costs. In addition, the indifference price depends on the current number of option's shares held by the marker maker, due to the nonlinearity of the utility function. By buying or selling options, the client changes the market maker's inventory, affecting the indifference price and, thus, generating price impact in the option's market. The expansion provided in this paper allows one to compute the price impact coefficient in the option's market (which is hard to measure directly, due to the lack of liquidity and/or transparency) in terms of the price impact coefficient $\eta$ in the underlying market (which is easier to measure directly), assuming the latter is small -- this connection is given explicitly by equation \eqref{eq.expansion.impact.deriv}.
Unlike the existing literature \cite{PST,MMS,EMK,BCE,KM,CHM,GW,BSh}, where the authors obtain expansions for the value function of the optimal hedging problem, to obtain the small impact expansion of the indifference price we need to expand a partial derivative of the value function. As the existing methods are not sufficient to obtain such an expansion, we employ a more direct approach that relies on the properties of the optimal control and on the stochastic representations of the derivatives of the value function, established in the preceding part of the paper.

To justify the interpretation of indifference price as the option price quoted by the market makers, we refer the reader to \cite{ENequil}. The latter is an online appendix to the present paper, which shows that the marginal utility indifference price is indeed an equilibrium price in a game with competing market makers, who trade dynamically in options (with a client) and hedge their positions by trading in the underlying.

The rest of the paper is organized as follows. In Section \ref{se:hedging}, we solve the problem of optimal hedging of a static position in the option.
This is done in several steps. First, we consider an approximation of the target stochastic control problem with the problems in which the state process contains additional noise and the controls are bounded. The latter features allow us to avoid the degeneracy and quadratic growth mentioned above and to characterize the solution of the approximating problem via the HJB equation. Then, using the martingale optimality principle, we derive a Forward-Backward Stochastic Differential Equation (FBSDE) for the optimal control of the approximating problem. Using the direct analysis of the original and the approximating control problems, we establish certain a priori estimates, which, along with the BSDE methods, allow us to obtain, in Theorem \ref{unifbound}, an upper bound on the absolute value of the optimal control that is uniform over the approximation parameters. Using the boundedness of the optimal control, we establish its feedback representation in Theorem \ref{theod0}, and the representation for indifference price follows easily from this result. In Section \ref{s.small}, we establish the asymptotic expansion of the indifference price (Theorem \ref{theo:expansion}), using the representations for the optimal control and for the derivatives of the value function established in Section \ref{se:hedging}.

\section{Optimal hedging strategy and the indifference price}
\label{se:hedging}
Let $(\Omega,\cF,(\cF_t)_{t\in[0,T]},\PP)$ be a filtered probability space where $(\cF_t)_{t\in [0,T]}$ is the augmentation of the filtration generated by the Brownian motions $W=(W_t)_{t\in[0,T]}$ and $B=(B_t)_{t\in[0,T]}$, where $B$ is only used for approximation purposes and is independent of $W$. 
Consider a (relatively) liquid financial market consisting of an adapted asset (stock) price process and a constant riskless asset.
In addition to the liquid market, we consider a European-type contingent claim with maturity $T$, written on the underlying liquid asset, and study the optimal investment problem of an individual agent with a static position in this option. The agent trades dynamically in the underlying creating linear temporary impact. More precisely, we assume that the liquid asset follows the Almgren-Chriss model with temporary price impact:
\begin{align}
&S_v = s + \int_t^v \sigma dW_r,\quad \pi_v=\pi + \int_t^v \nu_u du,\label{eq.rev.FinModel1}\\
&X_v = x - \int_t^v (\eta \nu_r + S_r)\nu_rdr = x - \eta \int_t^v \nu^2_r dr - \int_t^v S_r d\pi_r,\label{eq.rev.FinModel2}
\end{align}
where $\pi$ represents the inventory process, $S$ is the unaffected price of one unit of the asset, $X$ is the cash position of the agent, and $\nu$ is the rate with which the agent chooses to purchase the asset. The constant $\eta>0$ represents the linear temporary price impact of the agent on the underlying asset.
We denote the payoff of the contingent claim by $H(S_T)$, with a function $H:\RR\to\RR$.
The goal of this section is to find a tractable representation for the marginal utility indifference price of this option as well as for the value function and the optimal exponential-utility-based hedging strategy of this option.

\smallskip

As the HJB equation of a Merton problem in the Almgren-Chriss model is degenerate and has quadratic nonlinearity in the gradient (see the introduction), it is convenient to regularize the problem by adding an independent noise to the state processes and a cap on the admissible controls (this regularization will ultimately be removed).
Thus, within this section, we re-define the state processes $(\pi,S,X)$ as follows. For any initial condition $\pi,s,x\in\RR$, any $0\leq t\leq v<T$, and any $\delta,\epsilon\geq0$, we consider
\begin{align}
&S_v = s + \int_t^v \sigma dW_r,\label{eq.price}\\
&\pi_v=\pi + \int_t^v (\nu_u du + \delta dB_u),\label{eq.positiion}\\
&X_v = x - \eta \int_t^v \nu^2_r dr - \int_t^v S_r d\pi_r
= x - \eta \int_t^v \nu^2_r dr - \int_t^v S_r(\nu_rdr + \delta dB_r),\label{eq.wealth}
\end{align}
where $B$ and $W$ are two independent Brownian motions, and $\nu \in \cA^{\epsilon}(t,T)$ is the set of $\RR$-valued stochastic processes that are progressively measurable w.r.t. $\cF^t_r:=\sigma\{W_r-W_t,\1_{\{\d>0\}}(B_r-B_t):r\in [t,T]\}$ and are such that $|\nu|\leq 1/\epsilon$ and $\EE\left[\int_t^T\nu_r^2dr\right]<\infty$ (the latter, clearly, is only needed when $\epsilon=0$). 

The agent aims to maximize the expected utility of her terminal wealth:
\begin{align}\label{eq.repAgent.Obj.def}
\hat{V}^{\delta,\epsilon}(t,s,\pi,x,Q):=\sup_{\nu \in \cA^\e(t,T)} \EE\left[-\exp\left(-\gamma \left(X_T+\pi_T S_T - l\frac{\pi_T^2}{2} + Q H(S_T)\right)\right)\right],
\end{align}
with the dynamics of the state processes given by (\ref{eq.price})--(\ref{eq.wealth}).
We are mainly interested in the case $\delta=\e=0$, which turns the model \eqref{eq.positiion}--\eqref{eq.wealth} into \eqref{eq.rev.FinModel1}--\eqref{eq.rev.FinModel2} and lifts the artificial bound on the controls in \eqref{eq.repAgent.Obj.def}. The case of $\delta,\e>0$ is included for technical reasons, as a way of regularizing the problem.
\smallskip

Denote 
$$
P(t,s)=\EE_{t,s}[H(S_T)].
$$ 
We make the following assumption on $H$, which holds throughout this section, even if not referenced explicitly. 
\begin{ass}\label{assume:P}
$H$ is globally Lipschitz-continuous.
\end{ass}
Note that the above assumption implies that $P(t,s)\in C^{1,3}([0,T)\times\RR)\cap C([0,T]\times\RR)$, that $\partial_s P$ is absolutely bounded on $[0,T]\times\RR$, and that $P(t,\cdot)$ is linearly bounded, uniformly over $t\in[0,T]$.
Using $P$, we can write the terminal wealth generated by a strategy $\nu$ as 
\begin{align*} X_T+\pi_T S_T - l\frac{\pi_T^2}{2} + Q H(S_T)
=x+\pi s+QP(t,s)+\int_t^T (\pi_r+Q\pa_sP(r,S_r))dS_r -\eta \int_t^T \nu_r^2dr -\frac{l}{2}\pi_T^2.
\end{align*}
Using the above, it is easy to see that, for any $(t,s,\pi,x,Q)\in[0,T]\times\RR^4$, we have
\begin{align*}
&\EE\left[-\exp\left(-\gamma \left(X_T+\pi_T S_T - l\frac{\pi_T^2}{2} + Q H(S_T)\right)\right)\right]= -e^{-\gamma(x+\pi s+QP(t,s))}J(t,s,\pi,Q;\nu),
\end{align*}
where
\begin{align}
J^\d(t,s,\pi,Q;\nu)
&:= \EE\left[e^{\Psi^\d(t,\pi, \nu)+Q\Gamma(t,s) }\right],\label{eq.J.def}\\
\Gamma(t,s)
&:=-\gamma\left(P(T,S_T)-P(t,s) \right)=-\gamma\int_t^T \pa_s P(r,S_r)dS_r,\nonumber\\
\Psi^\d(t,\pi,\nu)&:=\gamma\eta \int_t^T \nu^{2}_rdr +\frac{\gamma l}{2}\left(\pi+\int_t^T(\nu_rdr + \delta dB_r)\right)^2
-\gamma\int_t^T \left(\pi+\int_t^r(\nu_ldl+\d dB_l)\right)dS_r.\notag
\end{align}
Note that $ \Gamma$ and $\Psi^\d$ are in fact random and depend on the paths of the two Brownian motions on $[t,T]$. The above yields
\begin{align}\label{eq:defbarv}
\hat V^{\d,\e}(t,s,\pi,x,Q)=-e^{-\gamma(x+\pi s+QP(t,s))}U^{\d,\e}(t,s,\pi,Q),
\end{align}
where
\begin{align}
U^{\d,\e}(t,s,\pi,Q)&:= \inf_{\nu\in\cA^{\e}(t,T)} J^\d(t,s,\pi,Q;\nu) =\inf_{\nu\in\cA^{\e}(t,T)} \EE\left[e^{\Psi^\d(t,\pi, \nu)+Q\Gamma(t,s) }\right].\label{def:U}
\end{align}
Note that $\Gamma$ does not depend on $\nu$, and, due to Assumption \ref{assume:P}, $\Gamma(t,s)$ is linearly bounded in $s$, uniformly over $t$.
For convenience, we also introduce
\begin{align}
u^{\d,\e}(t,s,\pi,Q)&:= \log U^{\d,\e}(t,s,\pi,Q) = \log\left(\inf_{\nu\in \cA^\e(t,T)} \EE\left[e^{\Psi^\d(t,\pi, \nu)+Q\Gamma(t,s) }\right]\right).\label{eq:defue}
\end{align}
As shown below, due to the presence of expectation of the exponential of a square of Brownian motion, we only prove the finiteness of $J^\d$, $U$, $u$, and $\hat V$, for $\d\geq 0$ small enough.

\subsection{PDE representation of the value function}

The following proposition provides the value of $u$ and, in turn, of $\hat{V}$, for the case with no price impact ($\eta=0$), no extra noise ($\delta=0$), and no constraints ($\e=0$).
Its proof follows easily from the fact that the payoff $H(S_T)$ can be replicated perfectly when $\eta=0$ and that the replication strategy can be approximated by the absolutely continuous ones, so that the associated objective values of the agent converge.

\begin{lemma}\label{prop:frictionless}
If $\eta=\delta=\e=0$, then, for all $t<T$, $(s,\pi,Q)\in\RR^3$, we have $u^{0,0}(t,s,\pi,Q)=0$ and $u^{0,0}(T,s,\pi,Q)=\frac{l\gamma \pi^2}{2}$.
\end{lemma}
\begin{remark}
Let us comment on the discontinuity of the value function in the case $\eta=\delta=\e=0$ and $l>0$.
In the absence of price impact, the agent's optimal strategy is to insure (at least along an approximating sequence) that $\pi_t\approx-Q\pa_sP(t,S_t)$ for $t<T$ and that $\pi_T=0$. This is possible if the agent starts at time $t<T$, as it costs her nothing to adjust her position in the underlying at an arbitrarily high rate. If the agent starts at $t=T$, there is simply no time left to trade, which yields $u^{0,0}(T,s,\pi,Q)=\frac{l\gamma \pi^2}{2}$. The next proposition shows that, in the presence of price impact (i.e., with $\eta>0$), the discontinuity in the value function disappears, in particular, because the agent can no longer liquidate her position in the underlying right before time $T$ at no cost. 
\end{remark}

Next, we return to the case $\eta>0$ and general $\d,\e\geq0$.
It is easy to see that the HJB equation for the value function \eqref{eq.repAgent.Obj.def} (derived heuristically) is:
\begin{align}
&\partial_t \hat V^{\delta,\epsilon} + \frac{\sigma^2}{2}\partial_{ss} \hat V^{\delta,\epsilon} + \frac{\d^2}{2}\partial_{\pi\pi} \hat V^{\delta,\epsilon} - \d^2 s\partial_{\pi x}\hat V^{\delta,\epsilon} + \frac{\d^2}{2}s^2\partial_{xx} \hat V^{\delta,\epsilon}
+\sup_{|\nu|\leq1/\e}\,[\nu\partial_\pi \hat V^{\delta,\epsilon}-\nu(s+ \eta \nu)\partial_x \hat V^{\delta,\epsilon}]=0,\label{eq.HJB.Vhat}\\
& \hat V^{\delta,\epsilon}(T,s,\pi,x) = -\exp\left(-\gamma \left(x+\pi s - l\frac{\pi^2}{2} + Q H(s)\right)\right).\label{eq.HJB.Vhat.termcond}
\end{align}
We denote its Hamiltonian by
\begin{equation}\label{eq.He.def.rev}
H_\e(p):=\inf_{|\nu|\leq \frac{1}{\e}} \{\gamma \eta \nu^2+p\nu \},\quad p\in\RR.
\end{equation}
Note that, for $\e=0$,
\begin{equation}\label{eq.visc.H0.def}
H_{0}(p) = - \frac{1}{4\eta \gamma}p^2.
\end{equation}

\begin{proposition}\label{prop:comp}
Let Assumption \ref{assume:P} hold and consider arbitrary  $T,\sigma,\gamma,\eta>0$ and $Q\in\RR$.
Then, there exist constants $\underline \gamma,\overline \gamma,\overline\d,C>0$ (depending only on $(T,\sigma,\gamma,\eta,Q)$), such that, for all $(t,s,s',\pi)\in [0,T]\times \RR^3$, all $\d\in[0,\overline\d]$, and all $\e\geq0$, we have
\begin{align}\label{growth:u}
\underline \gamma\pi^2\leq u^{\d,\e}(t,s,\pi,Q) \leq \overline \gamma\left(\frac{\pi^2}{2}+1\right),
\end{align}
and
\begin{align}\label{lip:u}
|u^{\d,\e}(t,s,\pi,Q)-u^{\d,\e}(t,s',\pi,Q)| \leq C|s-s'|.
\end{align}
Moreover, for all $\d\in[0,\overline\d]$ and $\e\geq0$, $u^{\d,\e}(\cdot,\cdot,\cdot,Q)$ is a (continuous on $[0,T]\times\RR^2$) viscosity solution of 
\begin{align}\label{eq:u}
&0=\pa_t u^{\d,\e} +\frac{\sigma^2}{2}\pa_{ss}u^{\d,\e} + \frac{\delta^2}{2}\pa_{\pi\pi}u^{\d,\e} 
+ H_\e(\pa_\pi u^{\d,\e}) + \frac{\delta^2}{2}(\pa_\pi u^{\d,\e})^2
+\frac{\sigma^2}{2}\left(\pa_s u^{\d,\e} -\gamma(\pi+Q \pa_s P(t,s))\right)^2,\notag\\
&u^{\d,\e}(T,s,\pi,Q)=\frac{l\gamma}{2}\pi^2.
\end{align}
In addition, if $\d\e=0$, the viscosity solution of \eqref{eq:u} is unique in the class of functions satisfying \eqref{growth:u}-\eqref{lip:u}; and if $\d\e>0$, then $u^{\d,\e}(\cdot,\cdot,\cdot,Q)\in C^{1,2}([0,T)\times\RR^2)\cap C([0,T]\times\RR^2)$.
\end{proposition}
\begin{proof}
We drop the dependence of the functions on $Q$, $\d$, and $\e$, unless it is necessary. 
By the Cauchy-Schwarz and Jensen inequalities we have that
\begin{align*}
\underline u(t,\pi)&:=\inf_{\nu\in\cA^0(t,T)} \gamma \eta \int_t^T \EE[\nu_r^2]dr+\frac{\gamma l}{2}\left(\pi+\int_t^T\EE[\nu_r]dr \right)^2\\
&\leq\inf_{\nu\in\cA^\e(t,T)} \EE\left[\gamma \eta \int_t^T \nu_r^2dr+\frac{\gamma l}{2}\left(\pi+\int_t^T(\nu_rdr + \delta dB_r)\right)^2\right]\\
&\leq u(t,s,\pi) \leq \frac{1}{2} \log\left( \EE\left[e^{2Q\Gamma(t,s)}\right]\right)+\frac{1}{2}\inf_{\nu\in\cA(t,T)} \log \EE\left[e^{2\Psi(t,\pi,\nu) }\right]\\
&:= \frac{1}{2} \log\left( \EE\left[e^{2Q\Gamma(t,s)}\right]\right)+\frac{1}{2}\overline u(t,\pi).
\end{align*}
It is a standard exercise to verify that 
$$\underline u(t,\pi):=\underline \gamma_t\frac{\pi^2}{2},$$
with $\underline \gamma_t$ being the solution to the Riccati equation 
$$\frac{\underline \gamma'_t }{2}+H_0(\underline \gamma_t)=0,\quad\underline\gamma_T={\gamma l}.$$
Indeed, the latter can be deduced from the fact that the proposed $\underline u$ is a classical solution to the associated HJB equation
$$-\pa_t \underline u -H_0(\pa_\pi \underline u)=0.$$
Note that $\underline \gamma_\cdot$ is bounded from below on $[0,T]$.
Next, we deduce by a standard computation that
$$
\overline u \leq \log \EE\left[e^{2\Psi(t,\pi,0) }\right] \leq \overline \gamma \left(\frac{\pi^2}{2}+1\right),
$$
for some constant $\overline\gamma>0$ and for all $\delta\in[0,\overline\d]$, where $\overline\d$ is chosen so that
$$
\EE\left[ \exp\left(\gamma (l/2+T) \overline\d \sup_{t\in[0,T]} B^2_t\right) \right]< \infty.
$$
In addition, extracting an exponential martingale and using the boundedness of $\pa_s P$, we can estimate 
\newline$\left( \log\left[ \EE\exp(2Q\Gamma(t,s))\right]\right)/2$ from above by a constant $C$.
Thus, we have proved (\ref{growth:u}).

\smallskip

To show the Lipschitz-continuity of $u$ in $s$, we first observe that
\begin{align*}
s\to Q\Gamma(t,s)=-\gamma Q \left(P(T,s+\sigma(W_T-W_t))-P(t,s)\right)
\end{align*}
is Lipschitz-continuous. Thus, 
\begin{align*}
u(t,s',\pi,Q)&=\log\left(\inf_\nu \EE\left[e^{\Psi(t,\pi, \nu)+Q\Gamma(t,s')  }\right]\right)\\
&\leq \log\left(\inf_\nu \EE\left[e^{L|Q||s-s'|+\Psi(t,\pi, \nu)+Q\Gamma(t,s)  }\right]\right)
= u(t,s,\pi)+L|Q||s-s'|,
\end{align*}
with some constant $L>0$ which only depends on $P$ and $\gamma$.
Interchanging $s$ and $s'$, we obtain the Lipschitz-continuity of $u$, stated in (\ref{lip:u}). 

\smallskip

It remains to show that $u$ solves \eqref{eq:u}. 
To this end, we apply \cite[Corollary 5.6]{BouchardTouzi}, which states that the lower- and upper-semicontinuous envelopes of $\hat{V}$ (defined in \eqref{eq:defbarv}) are, respectively, viscosity super- and sub-solutions to the associated HJB equation \eqref{eq.HJB.Vhat}--\eqref{eq.HJB.Vhat.termcond}.
Note that the assumption of Lipschitz-continuity of the coefficients of the controlled state process, stated at the beginning of Section 5 of \cite{BouchardTouzi}, is not satisfied herein, as the drift of $X$ in \eqref{eq.wealth} is a quadratic function of $\nu$. Nevertheless, the Lipschitz property of the coefficients is only used in \cite[Section 5]{BouchardTouzi} to verify a part of \cite[Assumption A]{BouchardTouzi}. For the reader's convenience, we state \cite[Assumption A]{BouchardTouzi}, adapted to the present setting, in Appendix A. Due to the very simple form of equations \eqref{eq.positiion}--\eqref{eq.wealth}, this assumption is easily verified without using the Lipschitz property of the coefficients.
Multiplying $\hat{V}$ by an exponential and taking a logarithmic transformation (to pass from $\hat{V}$ to $u$ via \eqref{eq:defbarv}--\eqref{eq:defue}), we conclude that the lower- and upper-semicontinuous envelopes of $u$ are, respectively, viscosity super- and sub-solutions to \eqref{eq:u}.

First, we analyze the case $\d\e>0$. Using the dominated convergence, it is easy to show that, for any sufficiently small $\d>0$, $J^\d(t,s,\pi,Q;\nu)$ is continuous in $(t,s,\pi)$, uniformly over $|\nu|\leq1/\e$. This implies the continuity of $U$ in $(t,s,\pi)$ and, in turn, the continuity of $\hat V$ in $(t,s,\pi,x)$. The latter yields (via \cite[Proposition 5.4]{BouchardTouzi}) the strong dynamic programming principle for $\hat V$ (i.e., `$V^*$' and `$\phi$' can be replaced by `$V$' in equations (3.1) and (3.2) of \cite{BouchardTouzi}), which reads as follows: for any stopping time $\tau$ with values in $[t,T]$, we have
\begin{equation}\label{eq.Vhat.DPP}
\hat V(t,s,\pi,x) = \sup_{\nu \in \cA(t,T)} \EE \,\hat{V}(\tau,S_\tau,\pi^\nu_\tau,X^\nu_\tau).
\end{equation}
Next, we change the variables introducing $v:= e^{-R_1(T-t)}\hat V$ and use \eqref{eq.HJB.Vhat} to derive the PDE for $v$. We restrict the domain of the latter equation to $(0,T)\times [-R_2,R_2]^3$ and equip it with the condition $v=e^{-R_1(T-t)}\hat V$ on the boundary of this domain (note that it is consistent with the terminal condition \eqref{eq.HJB.Vhat.termcond} due to continuity of $\hat V$). For sufficiently large $R_1$, the resulting boundary-value problem for $v$ falls within the scope of Theorem 3 in Section 6.4 of \cite{KrylovNonlin}, which yields the existence of its classical solution. Undoing the change of variables and applying the standard verification argument (for which we use \eqref{eq.Vhat.DPP}), we conclude that $e^{R_1(T-t)} v$ coincides with the value function $\hat V$. Multiplying by the appropriate exponential and taking logarithmic transformation (see \eqref{eq:defue}--\eqref{eq:defbarv}), we conclude that $u$ solves \eqref{eq:u} on $(0,T)\times [-R_2,R_2]^2$ (which suffices, as $R_2>0$ is arbitrary). For the aforementioned verification, we use \eqref{eq.Vhat.DPP}, as well as the fact that the feedback optimal control is given by
$$
\nu_t = \left(\frac{-\pa_\pi u(t,S_t,\pi^\nu_t)}{2\gamma\eta}\right)\vee (-1/\epsilon)\wedge (1/\epsilon),
$$
and that the associated SDE for $\pi^\nu$ has a solution.

For the case $\d\e=0$, we recall that the lower- and upper-semicontinuous envelopes of $\hat{V}$ are, respectively, viscosity super- and sub-solutions to \eqref{eq.HJB.Vhat}--\eqref{eq.HJB.Vhat.termcond}. Changing the variables, we deduce that the lower- and upper-semicontinuous envelopes of $u$ are, respectively, viscosity super- and sub-solutions to \eqref{eq:u}.
Thus, it suffices to prove a comparison principle for \eqref{eq:u}.
To this end, we fix $C>0$ and without loss of generality we establish the comparison principle in the class of functions satisfying \eqref{growth:u} and \eqref{lip:u} for this given constant. 
This part of the proof is based on the results of \cite{DL}. Denote, for $(p,X,Y)\in\RR^3$,
\begin{align*}
\tilde C&:=2C+\gamma |Q|\sup_{t,s}|\pa_sP(t,s)|,\\
G(t,s,\pi,p,X,Y)
=&\frac{\sigma^2}{2}\sup_{|\beta|\leq \tilde C}\bigg\{-p(-2 \beta+2\gamma \pi)\\
&-\left({\beta^2}+2\beta \gamma Q\pa_sP(t,s)-2\gamma \pi\gamma Q\pa_sP(t,s) -\gamma^2 \pi^2\right) +X + Y\frac{\delta^2}{\sigma^2}\bigg\}.
\end{align*}
Note that, if $|p|\leq 2C$, we have that 
$$G(t,s,\pi,p,X)=\frac{\sigma^2}{2}\left(p-\gamma(\pi+Q \pa_s P(t,s))\right)^2+ \sigma^2\frac{X}{2}+ \delta^2\frac{Y}{2}.$$
Note that we want to characterize $u$ as a viscosity solution of \eqref{eq:u}, and to verify this property one needs to replace the derivatives of $u$ with the elements of sub- and super-jets. It is clear that, if $u$ is $C$-Lipschitz-continuous in $s$, then its sub- and super-jets in $s$ are absolutely bounded by $C$. Thus, thanks to \eqref{lip:u} and the Definition of $G$, any viscosity sub- or super-solution to \eqref{eq:u}, satisfying \eqref{growth:u}--\eqref{lip:u}, is, respectively, a sub- or super-solution to the following PDE:
\begin{align}
&0=\pa_t u +H_\e(\pa_\pi u)+\frac{\delta^2}{2}|\pa_\pi u|^2+G(t,s,\pi,\pa_s u,\pa_{ss}u,\pa_{\pi\pi}u)\label{eq.visc.edzero.1}\\
&u(T,s,\pi)=\frac{\gamma l}{2}\pi^2\nonumber.
\end{align}

Next, we consider $\d=0$. Then, the above PDE satisfies all the assumptions of \cite[Theorem 2.1]{DL}\footnote{The assumptions of \cite[Theorem 2.1]{DL} include continuity of $G$ at $t\uparrow T$, which may not hold herein. However, a careful examination of the proof of \cite[Theorem 2.1]{DL} reveals that this assumptions is not needed and only the continuity on $[0,T)$ is used in the proof.}, hence, the comparison principe holds for this equation, which, in turn, yields the comparison principle for \eqref{eq:u} (in the desired class).

Finally, we consider $\e=0$. Then, in view of the explicit formula for $H_0$ (see \eqref{eq.visc.H0.def}), equation \eqref{eq.visc.edzero.1} transforms into
\begin{align*}
0&=\pa_t u +\left( - \frac{1}{4\eta \gamma}+
\frac{\delta^2}{2}\right)|\pa_\pi u|^2+G(t,s,\pi,\pa_s u,\pa_{ss}u,\pa_{\pi\pi}u)\\
&=: \pa_t u +\tilde{H}_0(\pa_\pi u)+G(t,s,\pi,\pa_s u,\pa_{ss}u,\pa_{\pi\pi}u),
\end{align*}
where
\begin{equation*}
\tilde{H}_{0}(p) :=\inf_{\nu\in\RR} \left\{\frac{\gamma \eta}{1-2\delta^2\gamma\eta} \nu^2+p\nu \right\},
\end{equation*}
and, by possibly decreasing $\overline\d$, we ensure that $1-2\delta^2\gamma\eta>0$.
The above PDE, again, falls within the setting of \cite[Theorem 2.1]{DL}, which yields the desired comparison principle for \eqref{eq:u}.
\qed
\end{proof}

The following corollary shows that $u^{0,0}$ is a limit of $u^{\d,\e}$ as $\d,\e\downarrow0$.

\begin{cor}\label{cor:ue}
For any sequences $\d_n\downarrow 0$ and $\e_n\downarrow 0$, $u^{\d_n,\e_n}$ converges to $u^{0,0}$ locally uniformly. 
\end{cor} 
\begin{proof}
Recall the definition of $H_\e$ in \eqref{eq.He.def.rev} and notice that, for any $p\in\RR$,
$$
H_\e(p)=\inf_{|\nu|\leq \frac{1}{\e}} \{\gamma \eta \nu^2+p\nu\} \rightarrow H_0(p),
$$
as $\e\downarrow0$. Thus, the generator of \eqref{eq:u} is continuous in $\d,\e\geq0$, and the stability of viscosity solutions (cf. \cite{Barles}) yields that $\liminf_{(s',\pi',n)\rightarrow(s,\pi,\infty)} u^{\d_n,\e_n}(s',\pi')$ and $\limsup_{(s',\pi',n)\rightarrow(s,\pi,\infty)} u^{\d_n,\e_n}(s',\pi')$ are, respectively, viscosity super- and sub-solutions to \eqref{eq:u} with $\d=\e=0$ (note that \eqref{growth:u} implies that these candidate super- and sub-solutions are well defined). As the comparison principle holds for the latter equation (see the proof of Proposition \ref{prop:comp}), we obtain the statement of the corollary.
\qed
\end{proof}


\subsection{Existence, uniqueness, and stability of the optimal control}

We begin with the existence and uniqueness of the optimal control.

\begin{lemma}\label{le:charOptCont.le1}
There exists $\overline\d>0$, such that, for any $(t,s,\pi,Q)\in[0,T]\times\RR^3$, any $\d\in[0,\overline\d]$, and any $\e>0$, there exists an optimizer $\nu^{*,t,s,\pi,Q,\d,\e}$ of \eqref{def:U}.
\end{lemma}
\begin{remark}
The main contribution of this lemma is for $\delta =0$, since for $\delta>0$ we can easily obtain a feedback control from the maximizer of the Hamiltonian. 
\end{remark}
\begin{proof} 
Fix $(t,s,\pi,Q,\d,\e)\in[0,T)\times\RR^3\times [0,\bar \d]\times (0,1)$ and thanks to the finiteness of the value \eqref{growth:u}, pick an optimizing sequence $\{\nu^{n}\}_n$ in $\cA^{\e}(t,T)$. 
On the probability space $[t,T]\times \Omega$ with measure $\frac{1}{T-t}\mbox{Lebesgue}\times \PP$, the family of random variables $(r,\omega)\mapsto \nu_r^n$ are uniformly bounded. 
Thus, we can use the the Komlos' lemma in \cite[Lemma 2.1]{BSV} and in \cite[Theorem A1.1]{delbaen1994general} to obtain that there exist
$\nu^{*,n}$ in the convex envelop of $\{\nu^k:k\geq n\}$ and a process $\nu^*$(defined for almost all $t$) so that $\{\nu^{*,n}\}$ converges $\frac{1}{T-t}\mbox{Lebesgue}\times \PP$-a.s. and in $L^1$ to $\nu^*$. The almost sure converges implies that $\nu^* \in\cA^{\e}(t,T)$ and the $L^1$ convergence and the uniform boundedness imply that $\{\nu^{*,n}\}$ converges to $\nu^*$ in $L^p$.
Note also that for all $p\geq 1$, we can take $\bar \delta>0$ small enough so that  
$$\sup_{\nu,\nu'\in\cA^{\e}(t,T)}\EE\left[e^{p|\Psi^\d(t,\pi, \nu)-\Psi^\d(t,\pi, \nu')| }\right] <\infty.$$
Additionally, due to the Lipschitz-continuity of $H$, $\EE[e^{pQ \Gamma(t,s)}]<\infty$. Thus, the boundedness of the controls and the dominated convergence theorem easily yield that the mapping
$$\nu\in\cA^{\e}(t,T)\mapsto J^\d(t,s,\pi,Q;\nu)$$ is continuous in $L^2([t,T]\times \Omega)$. 
Finally, the convexity of $\nu\in\cA^{\e}(t,T)\mapsto J^\d(t,s,\pi,Q;\nu)$ and the fact that $\nu^n$ (and therefore $\nu^{*,n}$) is an optimizing sequence yields that $\nu^*$ is an optimizer of \eqref{def:U}.
\qed
\end{proof}

\begin{lemma}\label{lem:uniqueopt}
For any $\d\geq0$, there exist locally bounded functions $C_1$ and $C_2$ mapping, respectively, $(t,s,\pi)\in[0,T]\times\RR^2$ and $(t,s,\pi,\e)\in[0,T]\times\RR^2\times(0,\infty)$ into $(0,\infty)$, such that, for a.e. $\omega$, the mapping 
$\cA^\e(t,T)\ni\nu\mapsto e^{\Psi^\d(t,\pi,\nu) }$ is $\iota$-strong convex in the topology of $ L^2[t,T]$, where
$$
\iota := e^{-C_1\delta^2\sup_{r\in[t,T]}(B_T-B_r)^2 - C_2\sup_{r\in[t,T]}|W_T-W_r| 
+ \gamma\sigma\d\int_t^T (B_r-B_t)dW_r}/C_2.
$$
\end{lemma}
\begin{proof}
A direct computation of the second order Frechet derivative $\pa_{\nu\nu}\Psi^\d$ of $\nu \mapsto\Psi^\d(t,\pi,\nu) $ yields
$$\pa_{\nu\nu}\Psi^\d(t,\pi,\nu)(\nu',\nu')=2\eta\gamma\int_t^T (\nu'_r)^2dr+\gamma l \left(\int_t^T \nu'_rdr\right)^2.$$
Therefore, 
\begin{align*}
\pa_{\nu\nu}\left(e^{\Psi^\d(t,\pi,\nu)}\right)(\nu',\nu')&\geq e^{\Psi^\d(t,\pi, \nu)}\left(2\eta\gamma\int_t^T (\nu'_r)^2dr+\gamma l \left(\int_t^T \nu'_rdr\right)^2\right)\\
&\geq 2\eta\gamma e^{\Psi^\d(t,\pi,\nu)}\int_t^T (\nu'_r)^2dr.
\end{align*}
The following lower bound completes the proof:
$$\inf_{\nu\in \cA^\e(t,T)}2\eta\gamma e^{\Psi^\d(t,\pi,\nu)}\geq 2\eta\gamma \exp\left(-\frac{\gamma \eta}{\e^{2}}-\frac{\gamma l}{2}(|\pi|+\frac{T-t}{\e} + \d (B_T-B_t))^2-\gamma \sigma|\pi||W_T-W_t|
\right.
$$
$$
\left.
-\frac{\gamma\sigma}{\e}\int_t^T |W_T-W_r|dr
+ \gamma\sigma\d\int_t^T (B_r-B_t)dW_r\right). $$
\qed
\end{proof}

\begin{cor}\label{cor:unique}
There exists $\overline\d>0$, such that, for any $\d\in[0,\overline\d]$, $\e>0$, and $(t,s,\pi,Q)\in[0,T]\times\RR^3$, the optimizer $\nu^{*,t,s,\pi,Q,\d,\e}$ of \eqref{def:U} is unique.
\end{cor}
\begin{proof}
Consider the mapping $\cA^\e(t,T)\ni \nu\mapsto J^\d(t,s,\pi,Q;\nu)\in\RR$, which is well defined for sufficiently small $\overline\d>0$.
Using Lemma \ref{lem:uniqueopt} and the strict positivity of $\iota \exp(Q\Gamma(t,s))$ (with $\iota$ defined in Lemma \ref{lem:uniqueopt}), it is easy to deduce the strict convexity of the above mapping.
The latter implies uniqueness of the optimizer.
\qed
\end{proof}

\medskip

Throughout the remainder of this section, we denote by $\nu^{*,t,s,\pi,Q,\d,\e}$ the optimizer of \eqref{eq:defue}.

\medskip

The following proposition establishes the stability of the optimal control w.r.t. the initial condition $(s,\pi,Q)$.

\begin{proposition}\label{le:charOptCont.stabOptCont}
There exists $\overline\d>0$, s.t., for any fixed $t\in[0,T]$ and $\e>0$, there exist locally Lipchitz functions $C_{1,t,\overline \d,\e}$ and $C_{2,t,\overline \d,\e}$, with $C_{2,t,\overline \d,\e}(s,s,\pi,\pi,Q,Q)=0$, such that for all $s,s',\pi,\pi',Q,Q'\in \RR^6$ and all $\d\in[0,\overline\d]$, 
\begin{align}
 \int_t^T{\EE\left|\nu^{*,t,s,\pi,Q,\d,\e}_r-\nu_r^{*,t,s',\pi',Q',\d,\e}\right|^2dr}\leq& C_{1,t,\overline \d,\e}(s,s',\pi,\pi',Q,Q')|U^{\d,\e}(t,s,\pi,Q)-U^{\d,\e}(t,s',\pi',Q')|\notag\\
&+C_{2,t,\overline \d,\e}(s,s',\pi,\pi',Q,Q')(|U^{\d,\e}(t,s',\pi',Q')|+1).\label{eq.stabCont.mainEst}
 \end{align}
 In particular $$\RR^3\ni(s,\pi,Q)\mapsto \nu^{*,t,s,\pi,Q,\d,\e}\in L^2([t,T]\times\Omega)$$ is continuous for $\d\in[0,\overline\d]$. 
\end{proposition}
\begin{proof}
We fix $(t,\d,\e)$ and drop the dependence on these variables when not needed.
First, we notice that there exists a constant $L>0$, s.t.
\begin{align*}
&e^{L|Q-Q'|+L(|Q|+|Q'|)|s-s'|}U(s,\pi,Q)\geq \EE\left[e^{\Psi(\pi, \nu^{*,s,\pi,Q})+Q'\Gamma(s') }\right]\\
& \geq  \EE\left[e^{\Psi(\pi', \nu^{*,s,\pi,Q})+Q'\Gamma(s') }\right] - \EE\left[|e^{\Psi(\pi, \nu^{*,s,\pi,Q})}-e^{\Psi(\pi', \nu^{*,s,\pi,Q}) }|e^{Q'\Gamma(s')}\right]\\
&\geq \EE\left[e^{\Psi(\pi', \nu^{*,s',\pi',Q'})+Q'\Gamma(s') }\right]+ \EE\left[\pa_\nu \left(e^{\Psi(\pi', \nu^{*,s',\pi',Q'})+Q'\Gamma(s') }\right) (\nu^{*,s,\pi,Q}-\nu^{*,s',\pi',Q'})\right]\\
& + \EE\left[ \iota \int_t^T (\nu_r^{*,s,\pi,Q}-\nu_r^{*,s',\pi',Q'})^2dr\right]
 - \EE\left[|e^{\Psi(\pi, \nu^{*,s,\pi,Q})}-e^{\Psi(\pi', \nu^{*,s,\pi,Q}) }| e^{Q'\Gamma(s')}\right],
\end{align*}
where $\iota$ is defined in Lemma \ref{lem:uniqueopt} and the last inequality in the above relies on the $\iota$-convexity of the mapping 
$\nu\mapsto e^{\Psi(t,\pi,\nu) }$.
Due to the optimality of $ \nu^{*,s',\pi',Q'}$ and the admissibility of $ \nu^{*,s,\pi,Q}$, for the problem with initial condition $(s',\pi',Q')$, we have 
\begin{equation}\label{eq.SN.concern.1}
\EE\left[\pa_\nu \left(e^{\Psi(\pi', \nu^{*,s',\pi',Q'})+Q'\Gamma(s') }\right) (\nu^{*,s,\pi,Q}-\nu^{*,s',\pi',Q'})\right]\geq 0.
\end{equation}
Therefore, recalling the definition of $U$ in \eqref{def:U}, we obtain
\begin{align*}
&e^{L|Q-Q'|+L(|Q|+|Q'|)|s-s'|}U(s,\pi,Q)-U(s',\pi',Q')+\sup_{\nu\in\cA^\e(t,T)}\EE\left[|e^{\Psi(\pi, \nu)}-e^{\Psi(\pi', \nu) }|e^{Q'\Gamma(s')}\right]\\
&\geq \EE\left[ \iota \int_t^T (\nu_r^{*,s,\pi,Q}-\nu_r^{*,s',\pi',Q'})^2 dr\right].
\end{align*}
The Cauchy-Schwartz inequality yields 
$$
\sup_{\nu\in\cA^\e(t,T)}\EE\left[|e^{\Psi(\pi, \nu)}-e^{\Psi(\pi', \nu) }|e^{Q'\Gamma(s')}\right]
\leq \sup_{\nu\in\cA^\e(t,T)}\left(\EE e^{2\Psi(\pi', \nu)+2Q'\Gamma(s')}\right)^{1/2} \left(\EE |e^{\chi_{\d,\e}|\pi-\pi'|}-1|^2\right)^{1/2}
$$
$$
\leq e^{C_1(\e)\left(1+(\pi')^2+(s')^2 + (Q')^2\right)} \left(\EE |e^{\chi_{\d,\e}|\pi-\pi'|}-1|^2\right)^{1/2},
$$
with $ \chi_{\d,\e}:= \gamma l(|\pi|+|\pi'|+C_2(\e)+2\delta|B_T-B_t|)/2 + \gamma\sigma|W_T-W_t|>0$, which has finite exponential moments.
It is easy to see that there exists a sufficiently small $\overline\d>0$, s.t. $\EE (1/\iota)<\infty$ for all $\d\in[0,\overline\d]$. Thus, using the reverse Holder's inequality and the above estimates, we obtain 
\begin{align*}
&\EE (1/\iota) \left(e^{L|Q-Q'|+L(|Q|+|Q'|)|s-s'|}U(s,\pi,Q)-U(s',\pi',Q')\right.\\
&\left.+e^{C_1(\e)\left(1+(\pi')^2+(s')^2 + (Q')^2\right)} \left(\EE |e^{\chi_{\d,\e}|\pi-\pi'|}-1|^2\right)^{1/2}\right)
\geq {\EE\|\nu^{*,s,\pi,Q}-\nu^{*,s',\pi',Q'})\|^2_{L^2}},
\end{align*}
and we easily identify $C_{1,t,\overline \d,\e}$ and $C_{2,t,\overline \d,\e}$ whose regularity is a direct consequence of the existence of (finite) exponential moments of $\chi_{\d,\e}$. 
The continuity of $\RR^3\ni(s,\pi,Q)\mapsto \nu^{*,t,s,\pi,Q,\d,\e}\in L^2([t,T]\times\Omega)$ is now a consequence of the continuity of $U$.
\qed
\end{proof}

\smallskip

Throughout the remainder of this section, we fix $\overline\d>0$ for which the conclusions of Propositions \ref{prop:comp} and \ref{le:charOptCont.stabOptCont}, Lemma \ref{le:charOptCont.le1}, and Corollary \ref{cor:unique}, hold.

\subsection{Sensitivities of the value function}

Our next goal is to analyze the regularity of the partial derivatives of $U^{\d,\e}$, and hence $u^{\d,\e}$, w.r.t. $(s,\pi,Q)$. We begin with $J^{\d}$.
For any $\d\in[0,\overline\d]$, $\e>0$, and $\nu\in\cA^\e(t,T)$, we use Fubini's theorem to deduce: 
\begin{align}
\pa_s J^\d(t,s,\pi,Q;\nu)&=Q\,\EE\left[\pa_s\Gamma(t,s)e^{\Psi^\d(t,\pi, \nu)+Q\Gamma(t,s) }\right]\nonumber\\
&=\gamma Q\,\EE\left[ \left(\pa_s P(T,s+\sigma(W_T-W_t))-\pa_sP(t,s)\right)e^{\Psi^\d(t,\pi, \nu)+Q\Gamma(t,s) }\right],\label{eq.Sens.Js.def}\\
\pa_\pi J^\d(t,s,\pi,Q;\nu)&=\EE\left[\pa_\pi\Psi^\d(t,\pi, \nu)e^{\Psi^\d(t,\pi, \nu)+Q\Gamma(t,s) }\right]\nonumber\\
&=\EE\left[\left({\gamma l}(\pi+\int_t^T\nu_rdr+\delta (B_T-B_t))-\gamma(S_T-s)\right)e^{\Psi^\d(t,\pi, \nu)+Q\Gamma(t,s) }\right],\label{eq.Sens.Jpi.def}\\
\pa_Q J^\d(t,s,\pi,Q;\nu)&=\EE\left[\Gamma(t,s) e^{\Psi^\d(t,\pi, \nu)+Q\Gamma(t,s) }\right].\label{eq.Sens.JQ.def}
\end{align}

We recall the definition of equidifferentiability given in \cite{MS}.
\begin{definition}
For any $\e\geq 0$ and $t\in [0,T)$, we call the family of functions $\{f(\cdot,\nu): \RR\to \RR\}$, for all $\nu\in \times \cA^\e(t,T)$, equidifferentiable at $x\in \RR$ if, as $x'\to x$, the limit of $(f(x,\nu)-f(x',\nu))/(x-x')$ exists uniformly in $\nu\in \cA^\e(t,T)$. The family is equidifferentiable on a set if it is equidifferentiable at any point of the set. 
\end{definition}

\begin{lemma}\label{le:sensValue.mainLemma}
For any $\d\in[0,\overline\d]$, any $\e>0$, and any $t\in[0,T]$, the family 
$$\left\{(\pa_s J^\d(t,\cdot;\nu),\pa_\pi J^\d(t,\cdot;\nu),\pa_Q J^\d(t,\cdot;\nu)):\nu\in\cA^{\e}(t,T)\right\}$$
is uniformly bounded and equidifferentiable in each of its variables $(s,\pi,Q)\in \RR^3$.
In addition, for any $(t,s_0,\pi_0,Q_0)\in[0,T]\times\RR^3$, any $\d\in[0,\overline\d]$, and any $\e>0$, the mapping 
$$(s,\pi,Q)\mapsto (\pa_s J^\d(t,s_0,\pi_0,Q_0;\nu^{*,t,s,\pi,Q,\d,\e}),\pa_\pi J^\d(t,s_0,\pi_0,Q_0;\nu^{*,t,s,\pi,Q,\d,\e}),\pa_Q J^\d(t,s_0,\pi_0,Q_0;\nu^{*,t,s,\pi,Q,\d,\e}))$$
is continuous.
\end{lemma}
\begin{proof}
The uniform boundedness of $(\pa_s J^\d,\pa_\pi J^\d,\pa_Q J^\d)$ follows by direct estimates.
Formally differentiating the expressions for $(\pa_s J^\d, \pa_\pi J^\d,\pa_Q J^\d)$, we represent all partial derivatives of these terms as expectations of the quantities of the form 
$$\chi_{\alpha ,\beta}(t,s,\pi,Q) e^{\Psi^\d(t,\pi, \nu)+Q\Gamma(t,s) } \mbox{ for }\alpha,\beta=s,\pi,Q,$$
for some random weights $\chi_{\alpha,\beta}$.
Using the boundedness of $\nu\in \cA^\e(t,T)$, the fact that $\d$ is small enough, and Fubini's theorem, we verify these formal derivations and show that the second order derivatives can be bounded locally uniformly in $(s,\pi,Q,\nu)$. Using the dominated convergence, we also deduce that the second order derivatives are continuous in $(s,\pi,Q,\nu)$.
This implies the equidifferentiability of $(\pa_s J^\d, \pa_\pi J^\d,\pa_Q J^\d)$.
Finally, the continuity of $(\pa_s J^\d, \pa_\pi J^\d,\pa_Q J^\d)$ in $\nu$ and Proposition \ref{le:charOptCont.stabOptCont} imply the second statement of the lemma.
\qed
\end{proof}

\smallskip

The above lemma and the general version of the Envelop Theorem given in \cite{MS} allow us to establish the existence and representation of the partial derivatives of $U^{\d,\e}$.

\begin{proposition}\label{cor:smoothness}
For any $\d\in[0,\overline\d]$, any $\e>0$, and any $t\in[0,T]$, $U^{\d,\e}(t,s,\pi,Q)$ is continuously differentiable in $(s,\pi,Q)\in\RR^3$, with 
\begin{align}
\pa_s U^{\d,\e}(t,s,\pi,Q)&=Q\,\EE\left[\pa_s\Gamma(t,s)e^{\Psi^\d(t,\pi, \nu^{*,t,s,\pi,Q,\d,\e})+Q\Gamma(t,s) }\right],\label{eq:sder}\\
\pa_\pi U^{\d,\e}(t,s,\pi,Q)&=\EE\left[\pa_\pi\Psi^\d(t,\pi, \nu^{*,t,s,\pi,Q,\d,\e})e^{\Psi^\d(t,\pi, \nu^{*,t,s,\pi,Q,\d,\e})+Q\Gamma(t,s) }\right],\label{eq:pider}\\
\pa_Q U^{\d,\e}(t,s,\pi,Q)&=\EE\left[\Gamma(t,s)e^{\Psi^\d(t,\pi,  \nu^{*,t,s,\pi,Q,\d,\e})+Q\Gamma(t,s) }\right].\label{eq:Qder}
\end{align}
The above partial derivatives are continuous in $(t,s,\pi,Q)\in[0,T)\times\RR^3$. Moreover, they are locally H\"older-continuous in $(s,\pi,Q)$, locally uniformly over $t\in[0,T)$ and $\d\in[0,\tilde\d]$, with some $\tilde\d\in(0,\overline\d]$. For $\pa_\pi U^{\d,\e}$ and $\pa_Q U^{\d,\e}$, the latter two statements hold with the interval $[0,T)$ replaced by $[0,T]$.
\end{proposition}
\begin{proof}
Lemma \ref{le:sensValue.mainLemma} and \cite[Theorem 3]{MS} imply the existence of partial derivatives of $U^{\d,\e}$ w.r.t. $s$, $\pi$, and $Q$, and the representations \eqref{eq:sder}--\eqref{eq:Qder}. Due to the fact that $\e>0$, the boundedness of the controls and an application of dominated convergence theorem shows that these partial derivatives are jointly continuous in $(s,\pi,Q)$. Hence, $U^{\d,\e}$ is continuously differentiable w.r.t. $(s,\pi,Q)$.

Using \eqref{eq.stabCont.mainEst} and the differentiability of $U^{\d,\e}$, we conclude that the mapping
$$\RR^3\ni(s,\pi,Q)\mapsto \nu^{*,t,s,\pi,Q,\d,\e}\in L^2([t,T]\times\Omega)$$ is locally 1/2-H\"older-continuous, uniformly over small enough $\d\geq0$.
The latter observation, the explicit form of $\Psi$, $\Gamma$, $\pa_\pi\Psi$, $\pa_s\Gamma$ (see \eqref{eq.Sens.Js.def}--\eqref{eq.Sens.JQ.def}), and the Cauchy-Schwartz inequality, imply the desired H\"older-continuity of the partial derivatives.
It is easy to see that the H\"older exponents and the associated coefficients are uniform over $t\in[0,T-\varepsilon]$ and $\d\in[0,\tilde\d]$, with some fixed $\tilde\d\in(0,\overline\d]$ and for arbitrary $\varepsilon\in(0,T)$. Since, for $\alpha=s,\pi,Q$, the function $\partial_\alpha U^{\d,\e}(t,\cdot,\cdot,\cdot)$ is continuous uniformly over $t\in[0,T-\varepsilon]$ and, for any $(s,\pi,Q)$, the function $U^{\d,\e}(\cdot,s,\pi,Q)$ is continuous on $[0,T-\varepsilon]$, it is a standard exercise to check (by contradiction) that $\partial_\alpha U^{\d,\e}$ is jointly continuous on $[0,T-\varepsilon]\times\RR^3$, for any $\varepsilon\in(0,T)$. It remains to notice that the only reason we excluded $t=T$ in the preceding arguments is the possible discontinuity of $\pa_s\Gamma$ at $t=T^-$. Since this term does not appear in $\partial_\alpha U^{\d,\e}$ for $\alpha=\pi,Q$, we conclude that the latter derivatives are continuous in $(t,s,\pi,Q)\in[0,T]\times\RR^3$ and H\"older-continuous in $(s,\pi,Q)$ uniformly over $t\in[0,T]$ and $\d\in[0,\tilde\d]$.
\qed
\end{proof}
\begin{remark}
Due to the presence of the exponent `$2$' in the left hand side of \eqref{eq.stabCont.mainEst}, at this stage, we cannot establish additional regularity of the derivatives of $U$ (such as the existence of the second order derivatives). Nevertheless, further regularity is shown in Corollary \ref{le:equil.le1}.
\end{remark}

\subsection{Feedback representation of the optimal control}

In this subsection, we first derive a FBSDE for the optimal control assuming $\d,\e>0$, and use this equation to establish a uniform absolute bound on the optimal control. Then, taking limits as $\d,\e\rightarrow0$, we obtain an Ordinary Differential Equation (ODE) for the optimal inventory in the underlying, with $\d=\e=0$. 
We suppress the dependence on $Q$ in many quantities appearing in this subsection, as $Q$ remains constant.

\smallskip

Before proceeding, we comment briefly on the measurability issues.
Thanks to Proposition \ref{cor:smoothness}, for $\d\in[0,\overline\d]$ and $\e>0$, $u^{\d,\e}$ is continuous in $(t,s,\pi)$ and continuously differentiable in $(s,\pi)$. Hence, $\pa_s u^{\d,\e}$ and $\pa_\pi u^{\d,\e}$ are Borel measurable in $(t,s,\pi)$. 
The progressive measurability of $(r,\omega)\mapsto \nu^{*,t,s,\pi,\d,\e}_r$ implies the progressive measurability of the optimal inventory in the underlying,
$$
(r,\omega)\mapsto \pi^{*,t,s,\pi,\d,\e}_r := \pi + \int_t^r (\nu^{*,t,s,\pi,\d,\e}_l dl + \delta dB_l).
$$
Thus, we conclude that $(r,\omega)\mapsto \pa_\alpha u^{\d,\e}\left(r,S^{t,s}_r,\pi^{*,t,s,\pi,\d,\e}_r\right)$, for $\alpha=s,\pi$, are progressively measurable, which allows us to define the relevant quantities below.
Finally, the continuity of the mapping $(s,\pi)\mapsto \nu^{*,t,s,\pi,\d,\e}\in \cA^\e(t,T)$ implies the progressive measurability of $(r,\omega,s,\pi)\mapsto \nu^{*,t,s,\pi,\d,\e}_r$.

\smallskip

We begin with the (one-sided) martingale optimality principle for $U^{\d,\e}$.

\begin{lemma}\label{lem:martchange}
For any $\d\in[0,\overline\d]$, any $\e>0$, and any $(t,s,\pi)\in [0,T]\times\RR^2$, the process $(M^{t,s,\pi,\d,\e}_l)_{l\in[t,T]}$, defined by
\begin{align}
M^{t,s,\pi,\d,\e}_l:= U^{\d,\e}(l,S^{t,s}_l,\pi^{*,t,s,\pi,\d,\e}_l) \exp&\left(\int_{t}^l \gamma \eta (\nu^{*,t,s,\pi,\d,\e}_r)^2dr\right.\nonumber\\
&\left.-\sigma \gamma \int_{t}^l (\pi^{*,t,s,\pi,\d,\e}_r+Q\pa_s P(r,S^{t,s}_r))dW_r\right),\label{eq.Mtspi.def}
\end{align}
is a martingale with the terminal value
\begin{equation}\label{eq.MT.de.def}
M^{t,s,\pi,\d,\e}_T = e^{\Psi^\d(t,\pi, \nu^{*,t,s,\pi,\d,\e})+Q\Gamma(t,s) }.
\end{equation}
\end{lemma}
\begin{proof}
Throughout this proof, we fix $\d\in[0,\overline\d]$ and $\e>0$, and drop these superscripts.
Due to \eqref{def:U}, we have $U(T,S^{t,s}_T,\pi^{*,t,s,\pi}_T) = \exp((\pi^{*,t,s,\pi}_T)^2 \gamma l/2)$.
Then, the fact that $M^{t,s,\pi}$ satisfies the desired terminal condition follows directly from the definitions of $\Psi$ and $\Gamma$ (preceding \eqref{def:U}).
It remains to show the martingale property. To this end, we claim that the optimal control is consistent (i.e. satisfies the flow property): for any $t\leq l \leq T$, a.s.
\begin{equation}\label{eq.optCont.flowProp}
\nu^{*,t,s,\pi}_r = \nu^{*,l,S^{t,s}_l,\pi^{*,t,s,\pi}_l}_r,\quad \text{a.e. }r\in[l,T].
\end{equation}
To prove this claim, we use the tower property and obtain, for any $(t,s,\pi)\in[0,T]\times\RR^2$, $\nu\in \mathcal{A}_\e(t,T)$, and with the associated $(S,\pi)=(S^{t,s},\pi^{t,s,\pi})$,
\begin{align*}
\EE e^{\Psi(t,\pi, \nu)+Q\Gamma(t,s) }
= \EE\left[ \exp\left(\int_{t}^l \gamma \eta \nu_r^2dr - \sigma \gamma \int_{t}^l (\pi_r+Q\pa_s P(r,S_r))dW_r\right)
\EE\left( e^{\Psi(l,\pi_l, \nu)+Q\Gamma(l,S_l)}\,\vert\,\mathcal{F}^t_l\right) \right]\\
= \EE\left[ \exp\left(\int_{t}^l \gamma \eta \nu_r^2dr -\sigma \gamma \int_{t}^l (\pi_r+Q\pa_s P(r,S_r))dW_r\right)\times\right.\\
\left.\EE\left( e^{\Psi\left(l,\pi', \nu\left(z_{[t,l]}\otimes(W-W_l,B-B_l)_{[l,T]}\right)_{[l,T]}\right)+Q\Gamma(l,s')}\right)_{s'=S_l,\,\pi'=\pi_l,\,z = (W-W_t,B-B_t)} \right]\\
\geq \EE\left[ \exp\left(\int_{t}^l \gamma \eta \nu_r^2dr -\sigma \gamma \int_{t}^l (\pi_r+Q\pa_s P(r,S_r))dW_r\right)\times\right.\\
\left.\EE\left( e^{\Psi\left(l,\pi', \nu(z_{[t,l]})\otimes \nu^{*,l,s',\pi'}_{[l,T]} \right)
+Q\Gamma(l,s')}\right)_{s'=S_l,\,\pi'=\pi_l,\,z=\nu} \right]
= \EE e^{\Psi\left(t,\pi, \nu_{[t,l]}\otimes \nu^{*,l,S_l,\pi_l}_{[l,T]}\right)+Q\Gamma(t,s) },
\end{align*}
where '$\otimes$' denotes the concatenation of paths, and we view the admissible controls as functions of Brownian increments on the associated time intervals.
The inequality between the left and the right hand sides of the above display implies that the objective of the optimization problem \eqref{def:U} will not increase if we modify $\nu^{*,t,s,\pi}$ on $[l,T]$ to be equal to the right hand side of \eqref{eq.optCont.flowProp}. Then, due to uniqueness of the optimal control with the initial condition $(s,\pi)$ at time $t$, \eqref{eq.optCont.flowProp} must hold.

The martingale property follows easily from \eqref{eq.optCont.flowProp}: for $t\leq l\leq T$,
\begin{align*}
&\EE \left(M^{t,s,\pi}_T \,\vert\,\mathcal{F}_l^t\right)
=\EE \left(\exp\left(\Psi(t,\pi, \nu^{*,t,s,\pi}_{[l,T]})+Q\Gamma(t,s) \right) \,\vert\,\mathcal{F}_l^t\right)\\
&= \exp\left(\int_{t}^l \gamma \eta (\nu^{*,t,s,\pi}_r)^2dr - \sigma \gamma \int_{t}^l (\pi^{*,t,s,\pi}_r+Q\pa_s P(r,S^{t,s}_r))dW_r\right)\\
&\phantom{??????????????????????????????????????}
\cdot\EE \left(\exp\left(\Psi(l,\pi^{*,t,s,\pi}_l, \nu^{*,t,s,\pi})+Q\Gamma(l,S^{t,s}_l) \right) \,\vert\,\mathcal{F}_l^t\right)\\
&= \exp\left(\int_{t}^l \gamma \eta (\nu^{*,t,s,\pi}_r)^2dr - \sigma \gamma \int_{t}^l (\pi^{*,t,s,\pi}_r+Q\pa_s P(r,S^{t,s}_r))dW_r\right)\\
&\phantom{??????????????????????????????????????}
\cdot\EE \left(\exp\left(\Psi(l,\pi^{*,t,s,\pi}_l, \nu^{*,l,S_l,\pi^{*,t,s,\pi}_l,\d,\e})+Q\Gamma(l,S^{t,s}_l) \right) \,\vert\,\mathcal{F}_l^t\right)\\
&= \exp\left(\int_{t}^l \gamma \eta (\nu^{*,t,s,\pi}_r)^2dr - \sigma \gamma \int_{t}^l (\pi^{*,t,s,\pi}_r+Q\pa_s P(r,S^{t,s}_r))dW_r\right)
U^{\d,\e}(l,S^{t,s}_l,\pi^{*,t,s,\pi}_l)
= M^{t,s,\pi}_l.
\end{align*}
\qed
\end{proof}

\medskip

In order to derive an FBSDE representation for the optimal control it is convenient to work under a different probability measure.
To construct such a measure, we will use the martingale $M^{t,s,\pi,\d,\e}$. However, in order to apply Girsanov's theorem, it is convenient to use an alternative representation of this martingale via 
\begin{equation}\label{eq.Zde.def}
\ZZ^{\d,\e}(t,s,\pi):=\sigma(\pa_s u^{\d,\e}(t,s,\pi)-\gamma(\pi+Q\pa_s P(t,s))),
\end{equation}
provided in the following lemma.

\begin{lemma}\label{propdefmeasure}
For any $\d\in(0,\overline\d]$, any $\e>0$, and any $(t,s,\pi)\in [0,T]\times\RR^2$, the continuous modification of the martingale $(M^{t,s,\pi,\d,\e}_l)_{l\in[t,T]}$ is given by
\begin{align*}
M^{t,s,\pi,\d,\e}_l = U^{\d,\e}(t,s,\pi) \exp&\left(\int_{t}^l \ZZ^{\d,\e}\left(r,S^{t,s}_r,\pi^{*,t,s,\pi,\d,\e}_r\right)dW_r+\delta \pa_\pi u^{\d,\e}(r,S^{t,s}_r,\pi^{*,t,s,\pi,\d,\e}_r)dB_r\right.\\
&\left. - \frac{1}{2} \int_{t}^l (\ZZ^{\d,\e}\left(r,S^{t,s}_r,\pi^{*,t,s,\pi,\d,\e}_r\right))^2+\delta^2(\pa_\pi u^{\d,\e}(r,S^{t,s}_r,\pi^{*,t,s,\pi,\d,\e}_r))^2dr\right).
\end{align*}
\end{lemma}
\begin{proof}
As a martingale on a Brownian filtration, $M^{t,s,\pi,\d,\e}$ has a continuous modification. Since it is also positive, it must have the representation 
$$
M^{t,s,\pi,\d,\e}_l = U^{\d,\e}(t,s,\pi)  \exp\left(\int_{t}^l \phi^W_rdW_r+\phi^B_rdB_r-\frac{1}{2}\int_{t}^l (\phi^W_r)^2+(\phi^B_r)^2dr\right),
$$
for some $\phi^W$ and $\phi^B$ that are almost surely square integrable in time.
Applying It\^o's formula to the above representation of $M^{t,s,\pi,\d,\e}_l$ (viewed as a process in $l\in[t,T]$) and to the right hand side of \eqref{eq.Mtspi.def}, and equating the martingale terms, we obtain:
$$\phi^W_r=\sigma\left( \pa_s u^{\d,\e}(r,S_r,\pi^{\d,\e}_r)- \gamma   (\pi_r^{\d,\e}+Q\pa_s P(r,S_r))\right)= \ZZ^{\d,\e}\left(r,S_r,\pi^{\d,\e}_r\right)\mbox{ and }\phi^B_r=\d \pa_\pi u^{\d,\e}(r,S_r,\pi^{\d,\e}_r).$$
To justify the application of It\^o's formula to $u^{\d,\e}$, we recall that the latter is $C^{1,2}$ for $\d>0$.
\qed
\end{proof}

\smallskip

Using the martingales defined in Lemma \ref{lem:martchange}, for any $(t,s,\pi)\in [0,T]\times\RR^2$ and $\d\in[0,\overline\d]$, $\e>0$, we introduce the probability measure $\QQ^{t,s,\pi,\d,\e}$ on $\mathcal{F}^t$:
\begin{align}\label{def:Q}
\frac{d\QQ^{t,s,\pi,\d,\e}}{d\PP}:=\frac{M^{t,s,\pi,\d,\e}_T}{U^{\d,\e}(t,s,\pi) },
\end{align}
so that 
\begin{align}\label{eq:deftw}
\tilde W^{t,s,\pi,\d,\e}_l :=W_l - W_t - \int_t^l \ZZ^{\d,\e}\left(r,S^{t,s}_r,\pi^{*,t,s,\pi,\d,\e}_r\right)dr\mbox{ and }\\
\tilde B^{t,s,\pi,\d,\e}_l:=B_l - B_t -\int_t^l \delta \pa_\pi u^{\d,\e}(r,S^{t,s}_r,\pi^{*,t,s,\pi,\d,\e}_r)dr\label{eq:deftb}
\end{align}
 are independent standard Brownian motions on $[t,T]$ under $\QQ^{t,s,\pi,\d,\e}$. 
For convenience, we will often drop some (or all) of the superscript $(t,s,\pi,\d,\e)$ in the notation for $\QQ$, $\tilde B$, and $\tilde W$, when it causes no confusion.

\medskip

We now derive a FBSDE characterization of the optimal control under $\QQ$, for $\d,\e>0$. For notational convenience, we introduce the truncation function 
$$\phi_\e(x)=\left(\e^{-1}\wedge(x)\right)\vee (-\e^{-1}),\quad x\in\RR.$$
Note that $\phi_\e$ is an odd function.

\begin{proposition}\label{theodelta}
Let us fix an arbitrary initial point $(t_0,s_0,\pi_0)\in [0,T]\times \RR^2$, and constants $\d\in(0,\overline\d]$ and $\e>0$.
Then, the associated optimal control has a continuous modification satisfying 
\begin{align}\label{eq.Thm1.nustar.def}
\nu^{*,t_0,s_0,\pi_0,\d,\e}_t =- \phi_\e(Y^1_t/(2\eta\gamma)),\quad Y^1_t := \pa_\pi  u^{\d,\e}(t,S^{t_0,s_0}_t,Y^2_t),
\quad Y^2_t := \pi^{*,t_0,s_0,\pi_0,\d,\e}_t,
\end{align}
and $(Y^1,Y^2)$ solve the following FBSDE on $[t_0,T]$:
\begin{align}
\label{BSDEpapi}
Y^1_t=&\gamma l Y^2_T- \int_{t}^T\left( \gamma l\d^2 Y^1_r+\gamma \sigma\ZZ^{\d,\e}(r,S^{t_0,s_0}_r,Y^2_r)\right)dr
- \int_t^T \tilde Z^W_r d\tilde W_r-\int_t^T \tilde Z^B_r d\tilde B_r,\\
Y^2_t=&\pi_0+\int_{t_0}^t \left( -\phi_\e(Y^1_r/(2\eta\gamma)) + \d^2 Y^1_r\right)dr+\d (\tilde B_t-\tilde B_{t_0}).\label{eq:forw}
\end{align}
\end{proposition}
\begin{remark}
It is important to note that we are not using BSDE tools to claim the existence of a solution for the above system. A solution exists by the existence of the optimizer. 
\end{remark}
\begin{proof}
For convenience, we drop the dependence on $(\d,\e)$ and denote
$$
(S^0,\pi^0):=(S^{t_0,s_0},\pi^{*,t_0,s_0,\pi_0}).
$$
The representations \eqref{eq.Thm1.nustar.def} and \eqref{eq:forw} follow from the fact that $u^{\d,\e}\in C^{1,2}$ and from the existence of an optimal control in a feedback form (see Proposition \ref{prop:comp} and its proof).

It remains to prove that \eqref{BSDEpapi} holds. Note that the latter BSDE is equivalent to the statement that
$$
\pa_\pi  u(t,S^{0}_t,\pi^{0}_t) - \int_{t_0}^t \gamma l\d^2\left( \pa_\pi u(r,S^{0}_r,\pi^{0}_r)+\gamma \sigma\ZZ(r,S^{0}_r,\pi^{0}_r)\right)dr
$$
is a $\QQ^{t_0,s_0,\pi_0}$-martingale, with the terminal condition
$$
\gamma l \pi^{0}_T- \int_{t_0}^T\left( \gamma l\d^2 \pa_\pi u(r,S^{0}_r,\pi^{0}_r)+\gamma \sigma\ZZ(r,S^{0}_r,\pi^{0}_r)\right)dr.
$$
The terminal condition holds due to the fact that $\pa_\pi u(t,s,\pi) \rightarrow \gamma l \pi = \pa_\pi u(T,s,\pi)$, as $t\rightarrow T$ (cf. Proposition \ref{cor:smoothness}).
To prove the martingale property, we notice that the representation \eqref{eq.Mtspi.def} and the consistency property \eqref{eq.optCont.flowProp} imply, for all $t_0\leq t\leq t_1\leq T$,
\begin{equation}\label{eq.FBSELemma.Mconsist}
M^{t,S^{0}_t,\pi^{0}_t}_T
= M^{t,S^{0}_t,\pi^{0}_t}_{t_1}\frac{M^{t_1,S^{t,S^{0}_t}_{t_1},\pi^{*,t,S^{0}_t,\pi^{0}_t}_{t_1}}_{T}}{U\left(t_1,S^{t,S^{0}_t}_{t_1},\pi^{*,t,S^{0}_t,\pi^{0}_t}_{t_1}\right)}.
\end{equation}
Due to \eqref{eq:pider}, we have
\begin{align*}
&{\pa_\pi U(t,S^{0}_t,\pi^{0}_t)}=\EE\left[M^{t,S^{0}_t,\pi^0_t}_T\left(\gamma l\left(\pi_t^0+\int_t^T \nu^{*,t,S_t^0,\pi_t^0,\d,\e}_r dr +\d(B_T-B_t)\right)-\gamma\sigma(W_T-W_t)\right)|\mathcal{F}^{t_0}_t\right].
\end{align*}
Then, splitting the integration domain into $[t,t_1]$ and $[t_1,T]$, and using \eqref{eq.FBSELemma.Mconsist}, the standard properties of conditional expectation, and the consistency property \eqref{eq.optCont.flowProp}, we obtain:
\begin{align*}
&{\pa_\pi U(t,S^{0}_t,\pi^{0}_t)}
=\EE\left[M^{t,S^{0}_t,\pi^{0}_t}_{t_1} \frac{M^{t_1,S^{t,S^{0}_t}_{t_1},\pi^{*,t,S^{0}_t,\pi^{0}_t}_{t_1}}_{T}}{U\left(t_1,S^{t,S^{0}_t}_{t_1},\pi^{*,t,S^{0}_t,\pi^{0}_t}_{t_1}\right)}\left(\gamma l\pi_T^{*,t_1,S^{t,S^{0}_t}_{t_1},\pi_{t_1}^{*,t,S^{0}_t,\pi_t^{0}}} 
- \gamma\sigma(W_T-W_{t_1})\right)|\mathcal{F}^{t_0}_t\right]\\
&\qquad\qquad\qquad-\EE\left[M^{t,S^{0}_t,\pi^{0}_t}_{t_1} \int_{t}^{t_1} \gamma l\d^2 \pa_\pi u\left(r,S^{t,S^{0}_t}_r,\pi_{r}^{*,t,S^{0}_t,\pi_t^{0}}\right)
+ \gamma\sigma \ZZ\left(r,S^{t,S^{0}_t}_r,\pi_{r}^{*,t,S^{0}_t,\pi_t^{0}}\right) dr|\mathcal{F}^{t_0}_t \right]\\
&\qquad=\EE\left[M^{t,S^{0}_t,\pi^{0}_t}_{t_1}\left( \frac{\pa_\pi U\left(t_1,S^{t,S^{0}_t}_{t_1},\pi_{t_1}^{*,t,S^{0}_t,\pi_t^{0}}\right)}{U\left(t_1,S^{t,S^{0}_t}_{t_1},\pi^{*,t,S^{0}_t,\pi^{0}_t}_{t_1}\right)}
- \int_{t}^{t_1} \gamma l\d^2 \pa_\pi u(r,S^{0}_r,\pi^{0}_r)
+ \gamma\sigma \ZZ^{\d,\e}(r,S^{0}_r,\pi^{0}_r)dr \right)|\mathcal{F}^{t_0}_t\right].
\end{align*}
Next, we notice that \eqref{eq.FBSELemma.Mconsist} implies
\begin{equation*}
M^{t,S^{0}_t,\pi^{0}_t}_{t_1}
= U\left(t_0,S^{0}_{t},\pi^{0}_{t}\right)
\frac{M^{t_0,s_0,\pi_0}_{t_1}}{M^{t_0,s_0,\pi_0}_{t}}.
\end{equation*}
Collecting the above, we obtain
\begin{align*}
&\pa_\pi u(t,S^{0}_t,\pi^{0}_t) 
= \EE^{\QQ^{t_0,s_0,\pi_0}}\left[ \pa_\pi u\left(t_1,S^{t,S^{0}_t}_{t_1},\pi_{t_1}^{*,t,S^{0}_t,\pi_t^{0}}\right)|\mathcal{F}^{t_0}_t\right]\\
&-\EE^{\QQ^{t_0,s_0,\pi_0}}_t\left[\int_{t}^{t_1} \left(\gamma l\d^2 \pa_\pi u(r,S^{0}_r,\pi^{0}_r)
+ \gamma\sigma \ZZ^{\d,\e}(r,S^{0}_r,\pi^{0}_r)dr\right)| \mathcal{F}^{t_0}_t\right],
\end{align*}
which yields the desired martingale property.
\qed
\end{proof}


\smallskip

\begin{remark}\label{rem:Q.consist}
It is easy to deduce from \eqref{eq.FBSELemma.Mconsist} and from the measurability properties discussed at the beginning of this subsection, that, for any $\cF^{t_0}_T$-measurable random variable $\xi$ and any $r\in [t_0,T]$,
\begin{align}\label{eq:constproba}
\EE^{\QQ^{t_0,s_0,\pi_0,\d,\e}}\left[\xi|\mathcal{F}^{t_0}_r\right] = \EE^{\QQ^{r,s,\pi,\d,\e}}\left[\xi|\mathcal{F}^{t_0}_r\right]\big\vert_{(s,\pi)=\left(S^{t_0,s_0}_r,\pi^{*,t_0,s_0,\pi_0,\d,\e}_r\right)}.
\end{align}
\end{remark}

\medskip

Next, we use the FBSDE representation in Proposition \ref{theodelta} to estimate the optimal control uniformly in $\e,\d$.
To ease the notation, we introduce 
$$g^{\d,\e}_r:=\sigma(\pa_s u^{\d,\e}(r,S_r,\pi^{\d,\e}_r)-\gamma Q\pa_s P(r,S_r)),$$
which is bounded, uniformly over $(\omega,t_0,s_0,\pi_0,\d,\e)$, due to Assumption \ref{assume:P} and the Lipschitz-continuity of $u^{\d,\e}$ in $s$.
Then, \eqref{BSDEpapi} can be written as
\begin{align}
Y^1_t=&\gamma l Y^2_T+\int_{t}^T \left(\gamma^2 \sigma^2 Y^2_r-\gamma l\d^2 Y^1_r-\gamma\sigma g^{\d,\e}_r\right)dr-\int_t^T \tilde Z^W_r d\tilde W_r-\int_t^T \tilde Z^B_r d\tilde B_r.\label{eq:systbsde2}
\end{align}
 
\begin{theorem}\label{unifbound}
There exist constants $\d_0,C>0$, such that
\begin{align}\label{eq:boundstrat}
\left|\pa_\pi u^{\d,\e}(t,S^{t_0,s_0}_t,\pi^{*,t_0,s_0,\pi_0,\d,\e}_t)\right|
\leq C\left(1+|\pi_0|+\d \sup_{t_0\leq r\leq t}| B_r- B_{t_0}|\right),\quad t\in[t_0,T],
\end{align}
for all $(t_0,s_0,\pi_0)\in [0,T]\times\RR^2$, $\e>0$, and $\d\in (0,\d_0]$.
\end{theorem} 
\begin{proof}
For notational simplicity, we drop the dependence of the processes on $\d,\e$.
By the classical BSDE estimates applied to \eqref{eq:systbsde2} we obtain: 
\begin{align}
&\EE^\QQ\left[\sup_{t_0\leq t\leq T}|Y^1_t|^2 +\int_{t_0}^T |\tilde Z^W_t|^2+|\tilde Z^B_t|^2dt\right] \leq C\EE^\QQ\left[(Y^2_T)^2+\int_{t_0}^T (Y^2_r)^2+g_r^2dr\right]\label{eq:bsdebound}
\end{align}
where, as a part of our standing convention, we have omitted the dependence of $\QQ$ on $(t_0,s_0,\pi_0,\d,\e)$.
Making use of \eqref{eq:forw}, we apply Ito's formula to $Y^1_t Y^2_t$ to obtain 
\begin{align*}
\gamma l (Y^2_T)^2=&Y^1_{t_0} \pi_0 +\int_{t_0}^T\left( -\gamma^2\sigma^2 (Y^2_r)^2 + \gamma l \d^2 Y^1_r Y^2_r
-Y^1_r\phi_\e\left(\frac{Y^1_r}{2\eta \gamma}\right)+\d^2 (Y^1_r)^2 +\gamma\sigma Y^2_r g_r+\d \tilde Z^B_r\right)dr\\
&+\int_{t_0}^T Y^2_r\tilde Z_r^Wd\tilde W_r+\int_{t_0}^T Y^1_r \d+Y^2_r\tilde Z_r^Bd\tilde B_r.
\end{align*}
Consider $\lambda>0$, which is to be determined. Note also that there exists $C_\lambda$ such that, for all $a,b\in \RR$, we have $ab\leq \lambda a^2+C_\lambda b^2$. 
Then, there exists a constant $C_{\gamma,\sigma}>0$, depending only on $\gamma$ and $\sigma$, such that the above equality and \eqref{eq:bsdebound} imply that, for all $\e>0$ and all small enough $\d>0$, 
$$
\EE\left[\gamma l (Y^2_T)^2 +\int_{t_0}^T\left(\gamma^2\sigma^2(Y^2_r)^2+Y^1_r\phi_\e\left(\frac{Y^1_r}{2\eta \gamma}\right)\right)dr\right]\leq Y^1_{t_0}\pi_0
$$
$$
+\EE\left[\int_{t_0}^T\left(\d^2 (Y^1_r)^2 + \gamma l \d^2 Y^1_r Y^2_r +\gamma\sigma Y^2_r g_r+\d \tilde Z^B_r\right)dr\right]
$$
$$
\leq \lambda (Y^1_{t_0})^2+C_\lambda \pi_0^2 + \EE\left[\int_{t_0}^T\left(\d^2 (1+ \gamma l/2) (Y^1_r)^2 +\frac{1}{3}\gamma^2\sigma^2 (Y^2_r)^2 +C_{\gamma,\sigma}g^2_r+\d \tilde Z^B_r\right)dr\right]
$$
$$
\leq C_\lambda \pi_0^2 +\d \EE\left[\int_{t_0}^T \tilde Z^B_rdr\right]
$$
$$
+ \EE\left[(\lambda C+\d(T-t_0))(Y^2_T)^2+\int_{t_0}^T \left(\left(\lambda C+\d(T-t_0)+\frac{1}{3}\gamma^2\sigma^2 \right)(Y^2_r)^2+(\lambda C+\d(T-t_0)+C_{\gamma,\sigma})g_r^2dr\right)\right]
$$
$$
\leq C_\lambda \pi_0^2 +\d (T-t_0)+ \EE\left[((\lambda+\d) C+\d(T-t_0))(Y^2_T)^2+\int_{t_0}^T \left((\lambda+\d) C+\d(T-t_0)+\frac{1}{3}\gamma^2\sigma^2 \right)(Y^2_r)^2 dr\right]
$$
$$
+ \EE\left[\int_{t_0}^T((\lambda+\d) C+\d(T-t_0)+C_{\gamma,\sigma})g_r^2dr\right].
$$
We now choose small enough $\lambda,\d_0>0$, so that for all $\d\in(0,\d_0]$ we have 
$$((\lambda+\d) C+\d(T-t_0))\leq \frac{\gamma l}{2}\mbox{ and }(\lambda+\d) C+\d(T-t_0)+\frac{1}{3}\gamma^2\sigma^2\leq \frac{1}{2}\gamma^2\sigma^2.$$
Then, the previous estimate implies that, for all $\e>0$ and $\d\in(0,\d_0]$, 
$$\EE\left[\gamma l (Y^2_T)^2 +\int_{t_0}^T \left(\gamma^2\sigma^2 (Y^2_r)^2+Y^1_r\phi_\e\left(\frac{Y^1_r}{2\eta \gamma}\right)\right)dr\right]\leq C_1\left(\pi_0^2  +\d+\EE\left[\int_{t_0}^Tg_r^2dr\right]\right).$$
As $g$ is absolutely bounded, the above inequality implies 
$$\EE\left[\gamma l (Y^2_T)^2 +\int_{t_0}^T\left( \gamma^2\sigma^2 (Y^2_r)^2+Y^1_r\phi_\e\left(\frac{Y^1_r}{2\eta \gamma}\right)\right)dr\right]\leq C_2\left(\pi_0^2 + 1\right).$$
The above estimate and \eqref{eq:bsdebound} yield
$$|Y^1_{t_0}|^2 \leq C_3(\pi_0^2+1)=C_3((Y^2_{t_0})^2+1).$$
Repeating the procedure for arbitrary $t\in[t_0,T]$ in place of $t_0$ (and taking conditional, as opposed regular, expectations), we obtain 
$$|Y^1_t|^2 \leq C_3(\pi_t^2+1)=C_3((Y^2_t)^2+1),\quad t\in[t_0,T].$$
Bringing back the superscript $(\d,\e)$, we deduce from the above estimate that 
$$\phi_\e(-Y^1_t/(2\eta\gamma)) +\delta^2 Y^1_t= C^{\d,\e}_t Y^2_t + \tilde C^{\d,\e}_t,\quad t\in[t_0,T],$$
with some progressively measurable bounded processes $C^{\d,\e}$ and $\tilde C^{\d,\e}$, that are uniformly bounded for $\e>0$ and $\d\in(0,\d_0]$. 
Using the above representation, we can write the solution to \eqref{eq:forw} as follows: 
$$Y^2_r=\pi_0 e^{\int_{t_0}^rC^{\d,\e}_u du} + \int_{t_0}^r e^{\int_{s}^rC^{\d,\e}_u du}(\tilde C^{\d,\e}_s ds + \d d\tilde B_s)$$
where the anticipating integral is to be understood as $ \int_{t_0}^r e^{\int_{s}^rC^{\d,\e}_u du}d\tilde B_s= e^{\int_{t_0}^rC^{\d,\e}_u du}\int_{t_0}^r e^{-\int_{t_0}^{s}C^{\d,\e}_u du}d\tilde B_s$.
Using the above, we can represent the solution to \eqref{eq:systbsde2} as
\begin{align*}
Y^{1}_t=&e^{-\gamma l \d^2(T-t)} \EE^{\QQ^{\d,\e}}_t\left[\gamma l\left(\pi_0 e^{\int_{t_0}^TC^{\d,\e}_udu}+\int_{t_0}^Te^{\int_{s}^TC^{\d,\e}_u du}(\tilde C^{\d,\e}_s ds + \d d\tilde B_s)\right)\right.\\
&\left.+\int_{t}^T e^{\gamma l \d^2(T-r)}\left(\gamma^2\sigma^2 \left(\pi_0 e^{\int_{t_0}^rC^{\d,\e}_u du}+\int_{t_0}^re^{\int_{s}^r C^{\d,\e}_u du}(\tilde C^{\d,\e}_s ds + \d d\tilde B_s)\right)-\gamma\sigma g^{\d,\e}_r\right) dr\right].\\
\end{align*}
Then, the uniform boundedness of the processes $C^{\d,\e},\tilde C^{\d,\e},g^{\d,\e}$, the identity
$$
\int_{t_0}^r e^{\int_{s}^rC^{\d,\e}_u du}d\tilde B_s=\tilde B_r-e^{\int_{t_0}^rC^{\d,\e}_u du}\tilde B_{t_0}+\int_{t_0}^r C^{\d,\e}_se^{\int_{s}^rC^{\d,\e}_u du}\tilde B_sds,
$$
and the fact that $\tilde B$ is a Brownian motion under $\QQ^{\d,\e}$, yield
\begin{align*}
|Y^{1}_t|\leq \hat C(1+|\pi_0|+\d \sup_{t_0\leq r\leq t}|\tilde B_r-\tilde B_{t_0}|),\quad t\in[t_0,T],
\end{align*}
with a constant $\hat C$ independent of $(t_0,s_0,\pi_0)\in[0,T]\times\RR^2$, $\e>0$, and $\d\in(0,\d_0]$.
Note that $\tilde B_r-\tilde B_{t_0}=B_r-B_{t_0}-\d\int_{t_0}^t Y^{1}_r dr$. Thus, 
\begin{align*}
\sup_{t_0\leq r\leq t}|Y^{1}_r|&\leq \hat C (1+|\pi_0|+\d \sup_{t_0\leq r\leq t}| B_r- B_{t_0}|+\delta T\sup_{t_0\leq r\leq t}|Y^{1}_r|),
\end{align*}
which yields the desired estimate.  
\qed
\end{proof}

\medskip

Next, we establish the monotonicity of the feedback optimal control function $\pa_\pi u^{\d,\e}$.

\begin{lemma}\label{prop:convexity}
For any $(t,s)\in[0,T]\times\RR$, $\e>0$, and $\d\in(0,\overline \d]$, the functions $U^{\d,\e}(t,s,\cdot)$ and $u^{\d,\e}(t,s,\cdot)$ are convex. 
\end{lemma}
\begin{proof}
We omit the dependence of the functions on $\e,\d$.
The convexity of $U$ is a direct consequence of the convexity of a square function and an exponential. 
Indeed for any $(\lambda,t,s,\pi_1,\pi_2)\in[0,1]\times [0,T]\times \RR^3$, and any optimizing sequences $(\nu^{i,k})_{i=1,2,\,k\in \NN}$ for the problem \eqref{def:U} started at $(t,s,\pi_i)$, we have the inequality 
\begin{align*}
U(t,s,\lambda \pi_1+(1-\lambda)\pi_2)&\leq J(t,s,\lambda \pi_1+(1-\lambda )\pi_2;\lambda \nu^{1,k}+(1-\lambda)\nu^{2,k})\\
&\leq \lambda J(t,s,\pi_1; \nu^{1,k})+(1-\lambda)J(t,s,\pi_2; \nu^{2,k}).
\end{align*}
Taking $k$ to $\infty$, this leads to the convexity of $U$ in $\pi$. 

In order to prove the convexity of $u$ we adapt the ideas in \cite[Section 4]{GK}. First, we define the measure of convexity
$$[0,T]\times \RR^3\ni(t,s,\pi_1,\pi_2)\mapsto C(t,s,\pi_1,\pi_2)=u(t,s,\pi_1)+u(t,s,\pi_2)-2u\left(t,s,\frac{\pi_1+\pi_2}{2}\right)$$
Due to continuity of $u$, it is convex in $\pi$ if and only if $C(t,s,\pi_1,\pi_2)\geq 0$ for all $(t,s,\pi_1,\pi_2)\in [0,T]\times \RR^3$.
Due to convexity of $u$ at the final time, we have that 
$$C(T,\cdot)\geq 0.$$
Denoting $\bar \pi=\frac{\pi_1+\pi_2}{2}$, we differentiate $C$ and use the PDE \eqref{eq:u}, to obtain
$$
\cL^C C:=\pa_t C+\frac{\sigma^2}{2}\pa_{ss} C+\frac{\delta^2}{2}\left(\pa_{\pi_1\pi_1}C+\pa_{\pi_1\pi_1}C+2\pa_{\pi_2\pi_1}C\right)
$$
$$
=\left(\pa_t u+\frac{\sigma^2}{2}\pa_{ss}u+\frac{\delta^2}{2}\pa_{\pi\pi}u\right)(\pi_1)+\left(\pa_t u+\frac{\sigma^2}{2}\pa_{ss}u+\frac{\delta^2}{2}\pa_{\pi\pi}u\right)(\pi_2)
$$
$$
-2\left(\pa_t u+\frac{\sigma^2}{2}\pa_{ss}u+\frac{\delta^2}{2}\pa_{\pi\pi}u\right)\left(\bar \pi\right)
=-H_\e(\pa_\pi u({\pi_1}))-H_\e(\pa_\pi u({\pi_2}))+2H_\e(\pa_\pi u({\bar \pi}))
$$
$$
-\frac{\delta^2}{2}\left((\pa_\pi u({\pi_1}))^2+(\pa_\pi u({\pi_2}))^2- 2(\pa_\pi u({\bar \pi}))^2\right)
$$
$$
 -\frac{\sigma^2}{2}\left(\left(\pa_s u ({\pi_1})-\gamma(\pi_1+Q \pa_s P)\right)^2+\left(\pa_s u ({\pi_2})-\gamma(\pi_2+Q \pa_s P)\right)^2-2\left(\pa_s u ({\bar \pi})-\gamma(\bar \pi+Q \pa_s P)\right)^2\right).
$$
As $H_{\e}$ is Lipschitz-continuous, we can define bounded continuous functions $A^{\d,\e}_i$, for $i=1,2$, such that
$$-H_\e(\pa_\pi u({\pi_i}))+H_\e(\pa_\pi u({\bar \pi}))=(\pa_\pi u({\pi_i})-\pa_\pi u({\bar \pi}))A^{\d,\e}_i=A^{\d,\e}_i \pa_{\pi_i}C.$$
Additionally, by direct computation, we have
\begin{align*}
&(\pa_\pi u({\pi_1}))^2+(\pa_\pi u({\pi_2}))^2- 2(\pa_\pi u({\bar \pi}))^2=(\pa_{\pi_1} C)^2+(\pa_{\pi_2} C)^2+2 \pa_\pi u({\bar \pi})(\pa_{\pi_1} C+\pa_{\pi_2} C)\mbox{ and}\\
&\left(\pa_s u ({\pi_1})-\gamma(\pi_1+Q \pa_s P)\right)^2+\left(\pa_s u ({\pi_2})-\gamma(\pi_2+Q \pa_s P)\right)^2-2\left(\pa_s u ({\bar \pi})-\gamma(\bar \pi+Q \pa_s P)\right)^2\\
&=2\left(\pa_s u ({\bar \pi})-\gamma(\bar \pi+Q \pa_s P)\right)\pa_s C \\
&+\left(\pa_s u ({\pi_1})-\pa_s u ({\bar\pi})-\frac{\gamma}{2}(\pi_1-\pi_2)\right)^2+\left(\pa_s u ({\pi_2})-\pa_s u ({\bar\pi})-\frac{\gamma}{2}(\pi_2-\pi_1)\right)^2.
\end{align*}
Therefore, 
\begin{align*}
\cL^C C=&A^{\d,\e}_1 \pa_{\pi_1}C +A^{\d,\e}_2\pa_{\pi_2}C-{\sigma^2}\left(\pa_s u ({\bar \pi})-\gamma(\bar \pi+Q \pa_s P)\right)\pa_s C \\
&-\frac{\delta^2}{2}\left((\pa_{\pi_1} C)^2+(\pa_{\pi_2} C)^2+2 \pa_\pi u({\bar \pi})(\pa_{\pi_1} C+\pa_{\pi_2} C)\right)\\
& -\frac{\sigma^2}{2}\left(\left(\pa_s u ({\pi_1})-\pa_s u ({\bar\pi})-\frac{\gamma}{2}(\pi_1-\pi_2)\right)^2+\left(\pa_s u ({\pi_2})-\pa_s u ({\bar\pi})-\frac{\gamma}{2}(\pi_2-\pi_1)\right)^2\right)\\
\leq &A^{\d,\e}_1 \pa_{\pi_1}C +A^{\d,\e}_2\pa_{\pi_2}C-{\sigma^2}\left(\pa_s u\left({\frac{\pi_1+\pi_2}{2}}\right)-\gamma\left({\frac{\pi_1+\pi_2}{2}}+Q \pa_s P\right)\right)\pa_s C \\
&-{\delta^2} \pa_\pi u\left({\frac{\pi_1+\pi_2}{2}}\right)(\pa_{\pi_1} C+\pa_{\pi_2} C).
\end{align*}
Thus, $C$ is a supersolution, of at most quadratic growth, of a linear parabolic equation.  Due to Theorem \ref{unifbound} and the boundedness of $\pa_s u$, the coefficients of the generator of this linear PDE have at most linear growth, which is sufficient to claim that $C\geq 0$ (e.g., via the Feynman-Kac formula).
\qed
\end{proof}


\medskip

Recall that the main goal of this subsection is to establish a tractable representation and the key properties of the optimal control for $\d=\e=0$, by taking limits as $\e,\d\downarrow0$.

\begin{theorem}\label{theod0}
There exists an affine function $1/\e_0:\RR^+\rightarrow(0,\infty)$, such that the following statements hold.
\begin{itemize}
\item For any $(t,s,\pi)\in[0,T]\times\RR^2$ and any $\e\in[0,\e_0(|\pi|)]$, the optimal control $\nu^{*,t,s,\pi,0,\e}$ has a modification that is a.s. continuous in time and absolutely bounded (a.s., uniformly in $t$) by $1/\e_0(|\pi|)$.
\item For any $t\in[0,T]$ and any $\e\geq0$, the mapping $(s,\pi)\mapsto u^{0,\e}(t,s,\pi)$ is continuously differentiable, $\pa_s u^{0,\e}$ is continuous on $[0,T)\times\RR^2$, and $\pa_\pi u^{0,\e}$ is continuous and linearly bounded in $\pi$ on $[0,T]\times\RR^2$.\footnote{Note that the case $\e>0$ is covered by Proposition \ref{cor:smoothness}.}
\item For any $(t,s,\pi)\in[0,T]\times\RR^2$ and $\e\in[0,\e_0(|\pi|)]$, the aforementioned modification of the optimal control is given by
\begin{align}\label{eq:optcont}
\nu^{*,t,s,\pi,0,\e}_r = -\frac{1}{2\eta\gamma} \pa_\pi u^{0,\e}\left(r,S^{t,s}_r,\pi^{*,t,s,\pi,0,\e}_r\right),
\end{align}
where $\pi^{*,t,s,\pi,0,\e}$ is the a.s. unique solution to the ODE
\begin{equation}\label{eq.optContZero.ODEforPi}
d\pi^{*,t,s,\pi,0,\e}_r = -\frac{1}{2\eta\gamma}\pa_\pi u^{0,\e}\left(r,S^{t,s}_r,\pi^{*,t,s,\pi,0,\e}_r\right) \,dr,
\quad \pi^{*,t,s,\pi,0,\e}_{t} = \pi.
\end{equation}
\item For any $(t,s,\pi)\in[0,T]\times\RR^2$ and any $\e\in(0,\e_0(|\pi|)]$, we have, a.s.,
$$
\lim_{\d\downarrow0}\sup_{r\in[t,T]} \left|\nu^{*,t,s,\pi,\d,\e}_r - \nu^{*,t,s,\pi,0,\e}_r\right|=0
= \lim_{\e'\downarrow0} \left|\nu^{*,t,s,\pi,0,\e'}_r - \nu^{*,t,s,\pi,0,0}_r\right|,
$$
where every optimal control is understood as its continuous modification.
\end{itemize}
\end{theorem}
\begin{proof}
We fix $(t_0,s_0,\pi_0)\in[0,T]\times\RR^2$, and, in most instances, drop the dependence on these variables.

First, we prove the statement of the theorem excluding the case $\e=0$.
Consider $\e>0$, a sequence $\d_n\downarrow 0$, and the associated $\pi^{*,\d_n,\e}$, satisfying \eqref{eq:forw}:
\begin{equation}\label{eq.optContZero.ODEforPi.truncated}
d\pi^{*,\d_n,\e}_t = \left[-\phi_\e\left(\pa_\pi u^{\d_n,\e}\left(t,S_t,\pi^{*,\d_n,\e}_t\right)/(2\eta\gamma)\right) + \d^2_n \pa_\pi u^{\d_n,\e}\left(t,S_t,\pi^{*,\d_n,\e}_t\right)\right] dt + \d_n d \tilde B_t.
\end{equation}
Due to Theorem \ref{unifbound}, for a.e. random outcome, the drift in the above ODE is absolutely bounded by a constant times $1 +|\pi_0|+\d_n \sup_{t_0\leq r\leq T}| B_r- B_{t_0}|$ (the same constant for all $n$). Additionally, thanks to \eqref{eq:deftb} and \eqref{eq:boundstrat},
$$
\d_n \tilde B_t=\delta_n \left(B_t- \int_{t_0}^t\delta_n \pa_\pi u^{\d_n,\e}(r,S^{t_0,S_0}_r,\pi^{*,t_0,s_0,\pi_0,\d_n,\e}_r)dr\right)
$$ 
a.s. converges to zero uniformly in $t$. Thus, the family of functions $\{t\mapsto \pi^{*,\d_n,\e}_t\}_n$ is relatively compact in the $C$-norm on $[t_0,T]$ (for a fixed random outcome). Hence, up to a subsequence, we can assume that it converges as $n\rightarrow\infty$. 
Next, we observe that Corollary \ref{cor:ue} and the H\"older-continuity of the partial derivatives of $U^{\d,\e}(t,\cdot)$, which is uniform over small enough $\d\geq0$ (see Proposition \ref{cor:smoothness}), imply that, for any $t\in[t_0,T]$, $\pa_\pi U^{\d_n,\e}(t,\cdot)\rightarrow \pa_\pi U^{0,\e}(t,\cdot)$ locally uniformly, as $n\rightarrow\infty$. The latter, in turn, implies that $\pa_\pi u^{\d_n,\e}(t,\cdot)\rightarrow \pa_\pi u^{0,\e}(t,\cdot)$ locally uniformly.
Then, using the dominated convergence, it is easy to see that the limit of $\{\pi^{*,\d_n,\e}_\cdot\}_n$ (for a fixed random outcome, along a subsequence), denoted $\hat\pi^\e$, satisfies \eqref{eq.optContZero.ODEforPi.truncated} with $\d_n$ replaced by zero. Recall also that $|\hat\pi^\e|$ is bounded by an affine function of $|\pi_0|$ (independent of anything else, including the random outcome and the choice of a subsequence), which we denote by $1/\e_0$. Hence, for $\e\in(0,\e_0(|\pi_0|)]$, $\phi_\e$ can be replaced by identity, and we conclude that $\hat\pi^\e$ satisfies \eqref{eq.optContZero.ODEforPi}. Proposition \ref{cor:smoothness} and Lemma \ref{prop:convexity} imply that $\pa_\pi u^{0,\e}$ is jointly measurable and continuously increasing in $\pi$ (the latter property is only established for $\d>0$, but it extends trivially to $\d=0$ by taking a limit, as above). Then, a combination of Caratheodory's existence theorem and \cite[Theorem 3.1]{C} implies that the solution to \eqref{eq.optContZero.ODEforPi} is unique. Thus, the limits along all subsequences of $\{\pi^{*,\d_n,\e}\}_n$ must be the same, and we conclude that this sequence converges a.s., uniformly in $t\in[t_0,T]$, to $\hat\pi^\e$, and that $\nu^{*,\d_n,\e}$ converges in the same way to 
$$
\hat\nu^\e_t := -\frac{1}{2\eta\gamma} \pa_\pi u^{0,\e}\left(t,S_t,\hat\pi^{\e}_t\right).
$$
It only remains to show that $\hat\nu^\e = \nu^{*,0,\e}$. The latter follows easily from the aforementioned convergence and from the continuity of $J^\d(t,s,\pi,Q;\nu)$ in $(\d,\nu)$, for uniformly bounded $\{\nu\}$ (see \eqref{eq.J.def}).

\smallskip

Next, we consider the case $\e=0$.
Recall that, for all $\e\in(0,\e_0(|\pi_0|)]$, we have $\nu^{*,0,\e} = \nu^{*,0,\e_0(|\pi_0|)}$, which implies $u^{0,\e}(t,s,\pi)=u^{0,\e_0(|\pi_0|)}(t,s,\pi)$ for all $\e\in(0,\e_0(|\pi_0|)]$, $t\in[t_0,T]$, $s\in\RR$, $|\pi|\leq |\pi_0|$.
The first consequence of this observation is the existence of
$$
\hat\nu := \lim_{\e\downarrow0} \nu^{*,0,\e},
\quad \hat\pi := \lim_{\e\downarrow0} \pi^{*,0,\e},
$$
and the absolute boundedness of both processes (a.s., uniformly in $t$).
The second consequence is the existence of
\begin{align}\label{eq:cvvalues}
\hat v := \lim_{\e\downarrow0} \pa_\pi u^{0,\e},
\end{align}
where the convergence holds uniformly on all compacts. Corollary \ref{cor:ue} implies that 
$$
\hat v = \pa_\pi u^{0,0},\quad \hat \nu = \nu^{*,0,0},
$$
and the dominated convergence shows that the statement of the theorem holds for $\e=0$.
\qed
\end{proof}

\smallskip

\begin{remark}
For $\pi$ restricted to a compact, there is in fact no need to take the limit in \eqref{eq:cvvalues} -- it suffices to consider small enough $\e>0$. 
Thus, an immediate corollary of Theorem \ref{theod0} is the following: $\pa_\alpha u^{0,0}(t,s,\pi) =\pa_\alpha u^{0,\e}(t,s,\pi)$ for all $\e\in(0,\e_0(|\pi|)]$ and $\alpha=s,\pi$.
\end{remark}

\medskip

Throughout the rest of the paper, we interpret the optimal control $\nu^{*,t,s,\pi,\d,\e}$, for $\d,\e>0$ and $\d=0,\e\geq0$, as its continuous modification (appearing in Proposition \ref{theodelta} and in Theorem \ref{theod0}).

\medskip

It is important to notice that Theorem \ref{theod0} implies that the solution of the stochastic control problems \eqref{eq.repAgent.Obj.def} and \eqref{def:U}, with $\e=0$, is the same as their solution for $\e>0$, provided the latter is sufficiently small.
This observation allows us to extend some of the results established earlier in this section for $\e>0$ to the case $\e=0$. For example, it is clear that the statement of Proposition \ref{cor:smoothness} holds for $\e=0$. The following corollary extends the statement of Proposition \ref{prop:comp}, showing that $\hat{V}^{\delta,\epsilon}$ is a classical solution to the associated HJB equation, even for $\d=\e=0$ (note that we could not prove this fact directly in Proposition \ref{prop:comp}, as this extension requires the conclusion of Theorem \ref{theod0}).

\begin{cor}\label{le:equil.le1}
Denote by $\hat V(t,s,\pi,x,Q)$ the value function of \eqref{eq.repAgent.Obj.def}, with $\d=\e=0$. Then, $\hat V\in C^{1,2,1,\infty,1}([0,T)\times\RR^4)\cap C([0,T]\times\RR^4)$, and it is a classical solution to \eqref{eq.HJB.Vhat}--\eqref{eq.HJB.Vhat.termcond}:
\begin{align}
&\partial_t \hat V + \frac{\sigma^2}{2}\partial_{ss} \hat V
+\sup_{\nu\in\RR}\,[\nu\partial_\pi \hat V - \nu(s+ \eta \nu)\partial_x \hat V]=0,\label{eq.equil.HJB.fictitious.1}\\
& \hat V(T,s,\pi,x,Q) = -\exp\left(-\gamma\left(x+\pi s - l\frac{\pi^2}{2} + Q H(s)\right)\right).\label{eq.equil.HJB.fictitious.1.Tcond}
\end{align}
\end{cor}
\begin{proof}
It is shown in the proof of Proposition \ref{prop:comp} that $\hat V$ is a continuous viscosity solution to \eqref{eq.equil.HJB.fictitious.1}--\eqref{eq.equil.HJB.fictitious.1.Tcond}. Let us show that $\hat V$ is in fact a classical solution of this equation.
We begin by noticing that the supremum in \eqref{eq.equil.HJB.fictitious.1} is attained at
$$
\nu= \frac{1}{2\eta}\left(\frac{\partial_\pi\hat{V}}{\partial_x\hat{V}} - s\right),
$$
and that
$$
\sup_{\nu\in\RR}\,[\nu\partial_\pi \hat V-\nu(s+ \eta \nu)\partial_x \hat V]
= \frac{1}{4\eta} \left(\frac{\partial_\pi\hat{V}}{\partial_x\hat{V}} - s\right)^2 \partial_x\hat V.
$$
Next, we deduce from \eqref{eq:defbarv}--\eqref{def:U}, \eqref{eq:pider}, and \eqref{growth:u}, that
\begin{equation}\label{eq.equil.barV.est}
\left(|\hat V| + |\partial_x\hat V| + |\partial_x\hat V|^{-1} + |\partial_{xx}\hat V| + |\partial_{\pi x}\hat V| + |\partial_{\pi}\hat V|\right)(t,s,\pi,x,Q)|
\leq C_1(Q) e^{C_2(Q)(|x|+|\pi s| + \pi^2)},
\end{equation}
with some locally bounded $C_1,C_2>0$.
Treating the nonlinear part of \eqref{eq.equil.HJB.fictitious.1} as a given source term, we notice that the latter is measurable and absolutely bounded by the right hand side of \eqref{eq.equil.barV.est}, with possibly different constants. Then, for any fixed $(\pi,x,Q)\in\RR^3$, the Feynman-Kac formula yields a classical solution $\bar V(\cdot,\cdot,\pi_0,x,Q)\in C^{1,2}([0,T)\times\RR)$ to the aforementioned modification of \eqref{eq.equil.HJB.fictitious.1}, satisfying
\begin{align}
&\bar V(t,s,\pi,x,Q) = \EE \left[\int_t^T \left( \frac{1}{4\eta} \left(\frac{\partial_\pi\hat{V}}{\partial_x\hat{V}} - s - \sigma W_{r-t}\right)^2 \partial_x\hat V\right)(r,s+\sigma W_{r-t},\pi,x,Q) dr \right.\nonumber\\ 
&\left.\phantom{\frac{\frac{1}{2}}{\frac{1}{2}}?????????????????????????????????????????}
+ \hat V(T,s+\sigma W_{T-t},\pi,x,Q)\right].\label{eq.equil.hatV.FK}
\end{align}
Indeed, using the fact that $\partial_\pi\hat V$, $ \partial_x\hat V$ and $\hat V$ are continuous in all variables (recall that Proposition \ref{cor:smoothness} has been extended to the case $\e=0$) and the explicit form of the Gaussian transition density, along with the growth estimate \eqref{eq.equil.barV.est} and Fubini's theorem, we can show that $\partial_t \bar V$ and $\partial_{ss} \bar V$ are well defined and continuous in $(t,s,\pi,x,Q)$.

Next, we recall the value function $\hat V^{\d,\e}$ of \eqref{eq:defbarv} for $\d,\e>0$. We fix an arbitrary $\pi_0\in\RR$ and choose $\e=\e_0(|\pi_0|)/2$, where $\e_0$ is defined in Theorem \ref{theod0}, so that the optimal control $\nu^*$ of the unconstrained hedging problem is absolutely bounded by $1/\e$, for all initial underlying inventories in an open neighborhood of $\pi_0$, and for all $(t,s,x,Q)\in[0,T]\times\RR^3$.
In particular,
\begin{equation}\label{eq.AddReg.rev.1}
\sup_{|\nu|\leq1/\e}\,[\nu\partial_\pi \hat V^{\d,\e}-\nu(s + \eta \nu)\partial_x \hat V^{\d,\e}]
=\sup_{\nu\in\RR}\,[\nu\partial_\pi \hat V^{\d,\e}-\nu(s + \eta \nu)\partial_x \hat V^{\d,\e}]
= \frac{1}{4\eta} \left(\frac{\partial_\pi\hat{V}^{\d,\e}}{\partial_x\hat{V}^{\d,\e}} - s\right)^2 \partial_x\hat V^{\d,\e},
\end{equation}
for all sufficiently small $\d>0$, all $(t,s,x,Q)\in[0,T]\times\RR^3$, and over an open neighborhood of $\pi_0$.

Recall that, as shown in the proof of Proposition \ref{prop:comp}, $\hat V^{\d,\e}$ is a classical solution to \eqref{eq.equil.HJB.fictitious.1}--\eqref{eq.equil.HJB.fictitious.1.Tcond}. In addition, it is easy to see that, for sufficiently small $\d>0$, \eqref{eq.equil.barV.est} holds with $\hat V^{\d,\e}$ in place of $\hat V$. Then, the Feynman-Kac and It\^o's formulas imply the following representation (for sufficiently small $\d>0$):
\begin{align}
&\hat V^{\d,\e}(t,s,\pi_0,x,Q) = \EE \left(\int_t^T \left( \d^2 s\partial_{\pi x}\hat V^{\d,\e} + \frac{\d^2}{2}s^2\partial_{xx} \hat V^{\d,\e}\right.\right.
\label{eq.equil.hatVde.FK}\\
&\left.\left.+\sup_{|\nu|\leq1/\e}\,[\nu\partial_\pi \hat V^{\d,\e}-\nu(s + \sigma W_{r-t} + \eta \nu)\partial_x \hat V^{\d,\e}]\right)(r,s+\sigma W_{r-t},\pi_0+\d B_{r-t},x,Q) dr\right.\nonumber\\ 
&\left.\phantom{\frac{\frac{1}{2}}{2}}+ \hat V(T,s+\sigma W_{T-t},\pi_0+\d B_{T-t},x,Q)\right).\nonumber
\end{align}
It follows from the proof of Theorem \ref{theod0} that, for any $t\in[0,T)$, as $\d\downarrow 0$, 
\begin{equation}\label{eq.equil.hV.bV.conv}
(\partial_{\pi x}\hat V^{\d,\e},\partial_{xx}\hat V^{\d,\e},\partial_{s}\hat V^{\d,\e},\partial_{\pi}\hat V^{\d,\e},\partial_{x}\hat V^{\d,\e},\hat V^{\d,\e})(t,\cdot)\rightarrow (\partial_{\pi x}\hat V,\partial_{xx}\hat V,\partial_{s}\hat V,\partial_{\pi}\hat V,\partial_{x}\hat V,\hat V)(t,\cdot),
\end{equation}
locally uniformly in $(s,x,Q)$ and over an open neighborhood of $\pi_0$. The above convergence, the equations \eqref{eq.equil.hatV.FK}, \eqref{eq.equil.hatVde.FK} and \eqref{eq.AddReg.rev.1}, and the dominated convergence theorem, yield
$$
\hat V^{\d,\e}(t,s,\pi_0,x,Q) \rightarrow \bar V(t,s,\pi_0,x,Q),
$$
as $\d\downarrow0$. Using \eqref{eq.equil.hV.bV.conv} again, we conclude that $\hat V = \bar V$.
In particular, we conclude that $\partial_t \hat V$ and $\partial_{ss} \hat V$ are well defined and continuous in $(t,s,\pi,x,Q)$. On the other hand, the extended Proposition \ref{cor:smoothness} yields the continuity of $\pa_\alpha \hat V$, for $\alpha=x,s,\pi,Q$. Hence, $\hat V\in C^{1,2,1,1,1}([0,T]\times\RR^4)$. The infinite differentiability of $\hat V$ in $x$ (and the continuity of each derivative) follows easily from \eqref{eq:defbarv}.
\qed
\end{proof}

\medskip

Our final goal in this subsection is to establish a convenient BSDE-type representation of the optimal control for $\d=\e=0$, which will be used in Section \ref{s.small}. To this end, we recall the probability measure $\QQ^{t,s,\pi,\d,\e}$, defined in \eqref{def:Q} for $\d\in [0,\bar \d]$ and $\e>0$. We define the probability measure $\QQ^{t,s,\pi,0,0}$ in the same way: 
\begin{align}
&\frac{d\QQ^{{t,s,\pi,0,0}}}{d\PP}
:= M^{t,s,\pi,0,0}_T,\label{def:Qe}\\
&M^{t,s,\pi,0,0}_l := \exp\left(\int_t^l \ZZ^{0,0}\left(r,S^{t,s}_r,\pi_r^{*,t,s,\pi,0,0}\right)dW_r
- \frac{1}{2}\int_t^l \left(\ZZ^{0,0}\left(r,S^{t,s}_r,\pi^{*,t,s,\pi,0,0}_r\right)\right)^2 dr\right),
\quad l\in[t,T],\\
& \ZZ^{0,0}(t,s,\pi)=\ZZ(t,s,\pi):=\sigma(\pa_s u^{0,0}(t,s,\pi)-\gamma(\pi+Q\pa_s P(t,s))).\label{def:ZZZ}
\end{align}
Indeed, estimate \eqref{lip:u} and Theorem \ref{theod0} imply that $\ZZ^{0,0}\left(r,S^{t,s}_r,\pi_r^{*,t,s,\pi,0,0}\right)$ is absolutely bounded, uniformly over $r\in[t,T]$, which, in particular, yields the martingale property of $M^{t,s,\pi,0,0}$. The following lemma shows that the latter process is a limit of $M^{t,s,\pi,\d,\e}$, defined in \eqref{eq.Mtspi.def} (recall also Lemma \ref{propdefmeasure}).

\begin{lemma}\label{le:M.conv}
For any $(t,s,\pi)\in[0,T]\times\RR^2$ and any $r\in[t,T]$, we have
$$
(L^2)\,\lim_{\e\downarrow0}\lim_{\d\downarrow0}M^{t,s,\pi,\d,\e}_r = M^{t,s,\pi,0,0}_r.
$$
\end{lemma}
\begin{proof}
First, we recall the representation in Lemma \ref{propdefmeasure} and notice that Theorem \ref{theod0} and estimates \eqref{lip:u}, \eqref{eq:boundstrat} imply the existence of a function $\d_0:\,[1,\infty)\rightarrow(0,\infty)$ s.t.
$$
\sup_{(\d,\e)\in\left((0,\d_0(p)]\times(0,\e_0(|\pi|)]\right)\cup\{(0,0)\}} \EE \left( M^{t,s,\pi,\d,\e}_r \right)^p <\infty,\quad \forall\,p\geq1.
$$
In addition, the last statement of Theorem \ref{theod0}, the estimates \eqref{lip:u}, \eqref{eq:boundstrat}, and the uniform boundedness of the optimal control (as in the first statement of Theorem \ref{theod0}), yield
\begin{align}
(L^4)&\,\lim_{\e\downarrow0}\lim_{\d\downarrow0}\sup_{r\in[t,T]}|\pi_r^{*,t,s,\pi,\d,\e} - \pi_r^{*,t,s,\pi,0,0}|
=0\label{eq.Mconv.eq1}\\
&= \lim_{\e\downarrow0}\lim_{\d\downarrow0} \sup_{r\in[t,T-\varepsilon]}| \pa_\alpha u^{\d,\e}\left(r,S^{t,s}_r,\pi_r^{*,t,s,\pi,\d,\e}\right) - \pa_\alpha u^{0,0}\left(r,S^{t,s}_r,\pi_r^{*,t,s,\pi,0,0}\right)|,\nonumber
\end{align}
for $\alpha=s,\pi$ and for any $\varepsilon\in(0,T-t)$.
To estimate $\EE\left(M^{t,s,\pi,\d,\e}_r - M^{t,s,\pi,0,0}_r\right)^2$ for $r\in[t,T)$, it suffices to apply the inequality
$$
|e^x - e^y| \leq C |x-y| (e^{x}+e^y),\quad x,y\in\RR
$$
(which holds for a sufficiently large constant $C>0$), along with the Cauchy-Schwartz inequality and It\^o's formula (applied to the fourth power of a Brownian integral, to compute its expectation). To cover the case $r=T$, it suffices to notice that $M^{t,s,\pi,\d,\e}_\cdot$ is $L^2$-continuous at $T$ uniformly over $\d\in[0,\d_1]$ and $\e\in[0,\e_1]$, with sufficiently small $\d_1,\e_1>0$.
\qed
\end{proof}

\smallskip

Using the above constructions and Lemma \ref{le:M.conv}, we can now derive the desired representation of the optimal control for $\d=\e=0$ via a conditional expectation. To this end, we define
\begin{align}
\kappa&:=\sqrt{\frac{\sigma^2 \gamma}{2\eta}},\quad m(t):=-\kappa+\frac{2\kappa}{1-\frac{l-\sqrt{2\sigma^2 \gamma\eta} }{l+\sqrt{2\sigma^2 \gamma\eta}}e^{-2\kappa(T-t)}}\label{eq.kappa.m.def}\\
&=\left\{
\begin{array}{ll}
{\kappa \coth\left( \kappa(T-t)+\frac{1}{2}\ln \left(\frac{l+\sqrt{2\gamma \sigma^2\eta}}{l-\sqrt{2\gamma \sigma^2\eta}}\right)\right)} & {\mbox{ if }l-\sqrt{2\gamma \sigma^2\eta}>0,}\\
{\kappa \tanh\left( \kappa(T-t)+\frac{1}{2}\ln\left(\frac{l+\sqrt{2\gamma \sigma^2\eta}}{-l+\sqrt{2\gamma \sigma^2\eta}}\right)\right)} & {\mbox{ if }l-\sqrt{2\gamma \sigma^2\eta}\leq 0,}
\end{array}
\right.\nonumber
\end{align}
and note that $m$ is the continuous (i.e. non-exploding) solution to the ODE
$$-m'(t)+m^2(t)=\kappa^2,\quad t\in[0,T],\quad m(T)=\frac{l}{2\eta}.$$
In addition, we define the (Borel measurable) function $R$ via
\begin{equation}\label{eq.R.def}
\pa_\pi u^{0,0}(t,s,\pi)= 2\gamma\eta\, m(t) \pi + e^{\int_0^t m(r)dr} R(t,s,\pi).
\end{equation}
Notice that, thanks to \eqref{eq:optcont}, the optimal control for $\d=\e=0$ can be expressed in a feedback form via $R$.

\begin{proposition}\label{lem:boundR}
The function $R$ is absolutely bounded on $[0,T]\times\RR^2$. Moreover, for any $(t_0,s_0,\pi_0)\in[0,T]\times\RR^2$, the process $R_t:=R(t,S^{t_0,s_0}_t,\pi^{*,t_0,s_0,\pi_0,0,0}_t)$, for $t\in[t_0,T]$, is continuous and satisfies:
\begin{align}
&R_t= -\gamma\sigma^2 \EE^{\QQ^{t_0,s_0,\pi_0,0,0}}_t\left[\int_t^Te^{-\int_0^rm(v)dv} \left(\pa_s u^{0,0}(r,S^{t_0,s_0}_r,\pi^{*,t_0,s_0,\pi_0,0,0}_r)-\gamma Q\pa_sP(r,S^{t_0,s_0}_r)\right)dr\right].\label{def:R}
\end{align}
\end{proposition}
\begin{proof}
First, we note that the right hand side of \eqref{def:R} is absolutely bounded by a constant. Taking $t=t_0$, we deduce the absolute boundedness of the function $R$. Thus, it only remains to establish \eqref{def:R}.

Recall that, as follows from Proposition \ref{theodelta}, for $\d\in(0,\overline\d]$ and $\e>0$, the process $V^{\d,\e}$,
$$
V^{\d,\e}_t:=\pa_\pi  u^{\d,\e}(t,S^{t_0,s_0}_t,\pi^{*,t_0,s_0,\pi_0,\d,\e}_t) - \int_{t_0}^t \gamma l\d^2 \pa_\pi u^{\d,\e}(r,S^{t_0,s_0}_r,\pi^{*,t_0,s_0,\pi_0,\d,\e}_r)+\gamma \sigma\ZZ^{\d,\e}(r,S^{t_0,s_0}_r,\pi^{*,t_0,s_0,\pi_0,\d,\e}_r)dr,
$$
is a $\QQ^{t_0,s_0,\pi_0,\d,\e}$-martingale on $[t_0,T]$. Recalling the definition of $\ZZ^{\d,\e}$ (see \eqref{eq.Zde.def}) and using \eqref{eq.Mconv.eq1}, we conclude that, for any $t\in[t_0,T)$, there exists
$$
V^{0,0}_t:= (L^2)\,\lim_{\e\downarrow0}\lim_{\d\downarrow0} V^{\d,\e}_t
= \pa_\pi  u^{0,0}(t,S^{t_0,s_0}_t,\pi^{*,t_0,s_0,\pi_0,0,0}_t) - \int_{t_0}^t \gamma \sigma\ZZ^{0,0}(r,S^{t_0,s_0}_r,\pi^{*,t_0,s_0,\pi_0,0,0}_r)dr.
$$
Using the above convergence and Lemma \ref{le:M.conv}, it is easy to deduce that $V^{0,0}$ is a bounded $\QQ^{t_0,s_0,\pi_0,0,0}$-martingale on $[t_0,T)$ and, in particular, can be extended to $[t_0,T]$. Since the filtration is Brownian, there exists a continuous modification of this martingale.
Using \eqref{eq.R.def}, we obtain
$$
R_t = e^{-\int_0^t m(r)dr}\left(V^{0,0}_t - 2\gamma\eta  m(t) \pi^{*,t_0,s_0,\pi_0,0,0}_t 
+ \int_{t_0}^t \gamma \sigma\ZZ^{0,0}(r,S^{t_0,s_0}_r,\pi^{*,t_0,s_0,\pi_0,0,0}_r)dr\right).
$$
Then, by straightforward computations, we have
\begin{align*}
dR_t &= e^{-\int_0^t m(r)dr}dV^{0,0}_t - e^{-\int_0^t m(r)dr}V^{0,0}_t m(t) dt + 2\gamma\eta e^{-\int_0^t m(r)dr} m^2(t) \pi^{*,t_0,s_0,\pi_0,0,0}_t dt\\
 &= e^{-\int_0^t m(r)dr}dV^{0,0}_t
+ e^{-\int_0^t m(r)dr} \gamma \sigma^2\left(\pa_s u^{0,0}(t,S^{t_0,s_0}_t,\pi^{*,t_0,s_0,\pi_0,0,0}_t)-\gamma Q\pa_sP(t,S^{t_0,s_0}_t)\right)dt.
\end{align*}
Recalling that $V^{0,0}$ is a bounded continuous martingale under $\QQ^{t_0,s_0,\pi_0,0,0}$
and that
$$R(T,s,\pi)=(\pa_\pi u^{0,0}(T,s,\pi)- 2\gamma\eta\, m(T) \pi )e^{-\int_0^T m(r)dr} =e^{-\int_0^T m(r)dr}(\gamma l \pi-\gamma l \pi)=0,$$
 we conclude the proof.
\qed
\end{proof}

\begin{remark}
The representation \eqref{def:R} is to be compared to \cite[Theorem 3.1]{BSV} where the authors study a linear-quadratic optimization problem with price impact. Due to the local structure of the optimization objective, they are able to explicitly find the optimal strategy of the investor which consists in following a convolution of the future target position with an explicit kernel. In the exponential utility framework, considered herein, the problem is not linear-quadratic anymore. However, \eqref{def:R} indicates that the investor follows a similar convolution of the target position $-Q\pa_s P_t$ shifted with $\pa_s u$. The presence of $\pa_s u$ means that this equality does not provide an explicit solution to the optimization problem. However, the representation \eqref{def:R} allows us to control the effect of $\pa_s u$, and in Section \ref{s.small} we show that the impact of $\pa_s u$ can be controlled for small $\eta$, without decreasing the objective value at the main order of accuracy.
\end{remark}

\subsection{Utility indifference price}

Recall the definition of utility indifference price (cf. \cite{Carmona} and references therein).
\begin{definition}\label{def:indif.price}
For any initial condition $(s,\pi,x,Q)\in\RR^4$ at time $t\in[0,T]$, and any purchase quantity of the option $\Delta Q\in\RR$, the number $P^*(t,s,\pi,x,Q,\Delta Q)$ is the utility indifference price of $\Delta Q$ units of the option with payoff $H(S_T)$ if
$$
\hat{V}^{0,0}(t,s,\pi,x,Q) = \hat{V}^{0,0}(t,s,\pi,x + P^*(t,s,\pi,x,Q,\Delta Q),Q-\Delta Q),
$$
where $\hat V$ is defined in \eqref{eq.repAgent.Obj.def}.
\end{definition}

Recall that the utility indifference price is a natural notion of price in the options' markets. In particular, since the objective in \eqref{eq.repAgent.Obj.def} is nondecreasing in the option's payoff, the resulting utility indifference price is monotone w.r.t. the payoff: i.e., if one payoff function dominates another one from above everywhere, the price of the former is higher than the price of the latter. This, in turn, implies that the utility indifference price is free of static arbitrage: i.e., the price per unit is always between the lower and the upper bounds of $H$.\footnote{Note that the aforementioned monotonicity of the objective fails in the hedging problems with linear-quadratic objectives, causing potential static arbitrages in the indifference prices produced by such models.}

\smallskip

In view of \eqref{eq:defue} and \eqref{eq:defbarv}, we have
$$
P^*(t,s,\pi,x,Q,\Delta Q) = \Delta Q P(t,s) - \frac{1}{\gamma} (u^{0,0}(t,s,\pi,Q)-u^{0,0}(t,s,\pi,Q-\Delta Q)),
$$
where we bring back the dependence on $Q$ in related quantities.
To reduce the number of variables, we can assume that the option is traded in small quantity (at any fixed time). Then, we only need to study the marginal utility indifference price (also known as Davis price, see \cite{Davis}, \cite{Kramkov}, and references therein), which is defined as
\begin{align*}
p^*(t,s,\pi,Q)&:= \lim_{\Delta Q\rightarrow0} \frac{P^*(t,s,\pi,x,Q,\Delta Q)}{\Delta Q}
=P(t,s)-\frac{1}{\gamma}\pa_Q u^{0,0}(t,s, \pi,Q)=\EE^{\QQ^{t,s,\pi,Q,0,0}}[H(S_T)],
\end{align*}
where the last equality follows from \eqref{eq:Qder} (which is valid for $\e=0$ in view of the first statement of Theorem \ref{theod0}) and the fact that
\begin{equation}\label{eq.Q0T}
\frac{d\QQ^{t,s,\pi,Q,0,0}}{d\PP}
=\frac{e^{\Psi^\d(t,\pi,\nu^{*,t,s,\pi,Q,0,0})+Q\Gamma(t,s)}}{U^{0,0}(t,s,\pi,Q)}.
\end{equation}
The latter fact follows from \eqref{eq.MT.de.def}, Lemma \ref{le:M.conv}, and the last statement of Theorem \ref{theod0}.
Thus, the equilibrium price is the expectation under an equivalent probability measure, similar to the classical theory. 
We note that this measure depends on the claim and on the current positions of the agent in both the option and the underlying. 
Thanks to the definition of $\QQ^{t,s,\pi,Q,0,0}$, we also have 
\begin{align}
&p^*(t,s,\pi,Q)=\EE^{\QQ^{t,s,\pi,Q,0,0}}[H(S_T)]=P(t,s)\label{eq.expansion.pstar.rep1}\\
&+\sigma \EE\left[e^{\int_t^T \ZZ^{0,0}\left(r,S^{t,s}_r,\pi^{*,t,s,\pi,Q,0,0}_r,Q\right)dW_r-\frac{1}{2} \int_t^T \left(\ZZ^{0,0}\left(r,S^{t,s}_r,\pi^{*,t,s,\pi,Q,0,0}_r,Q\right)\right)^2dr} \int_t^T \pa_s P(r,S^{t,s}_r)dW_r\right],\notag
\end{align}
where $\ZZ^{0,0}(t,s,\pi,Q)=\sigma(\pa_s u^{0,0}(t,s,\pi,Q)-\gamma(\pi+Q\pa_s P(t,s)))$ is the function defined in \eqref{def:ZZZ}.


\section{Small impact expansion}
\label{s.small}

In the previous section, we have established various theoretical properties of the (log-) value function $u$, the optimal hedging strategy, and the marginal utility indifference price $p^*$, for an option with payoff $QH(S_T)$ in the Almgren-Chriss model. We have also derived useful representations for these quantities, which, in particular, allow for numerical approximations (see e.g. \eqref{eq:u}). 
However, the explicit expressions, that would provide additional insights into the behavior of $u$ and $p^*$, are not available. In this section, we derive an explicit expansion for $p^*$ assuming $\eta\to 0$. Note that, for $\eta=0$, the underlying market turns into the complete Bachelier model, where the option can be hedged perfectly by the standard delta-hedging strategy, and the marginal utility indifference price (as well as any reasonable notion of price) of the option is given by $P(t,s)$. Naturally, we would like to find the leading order of the difference between $P(t,s)$ and $p^*$ as $\eta\to 0$.

\smallskip

First, we make an additional modeling convention.
Namely, we claim that it is important to rescale the penalty coefficient for non-liquidation, $l$, appearing in \eqref{eq.repAgent.Obj.def}. Indeed, this coefficient is meant to reflect the losses associated with liquidating the remaining inventory in the underlying. The latter losses are due to the presence of price impact in the underlying market, hence, they should vanish as $\eta\to0$. Thus, in this section we make the following convention:
\begin{equation}\label{rescalingliquidation}
l = \bar l \eta,
\end{equation}
for some $\bar l\geq0$.
This convention implies that we should replace $l$ by $\bar l \eta$ in the formulas established in the previous section. In particular, since $\eta$ is small, the function $m$ defined in \eqref{eq.kappa.m.def} satisfies, in the new notation:
\begin{align*}
m(t)=\kappa \tanh\left( \kappa(T-t)+\frac{1}{2}\ln\left(\frac{\bar l\sqrt{\eta}+\sqrt{2\gamma \sigma^2}}{-\bar l\sqrt{\eta}+\sqrt{2\gamma \sigma^2}}\right)\right),
\end{align*}
and it solves 
$$-m'(t)+m^2(t)=\kappa^2,\, m(T)=\frac{\bar l}{2},$$
with $\kappa$ defined in \eqref{eq.kappa.m.def}. Note that $\eta\to0$ is equivalent to $\kappa\to\infty$.

\smallskip

For convenience, we often drop the superscript `$(\d,\e)$', as we mostly consider $\d=\e=0$ in this section (whenever this is not the case, the superscripts will appear).
In addition, to simplify the derivations, we will often omit the dependence on the initial condition $(s,\pi,Q)$, when it causes no confusion, and introduce
$$
P_t:=P(t,S_t),\quad \pa_s P_t:=\pa_s P(t,S_t).
$$

\medskip

Before proceeding, we need to establish a BSDE representation for $\pa_s u$, which is similar to \eqref{eq.R.def}--\eqref{def:R} established for $\pa_\pi u$. To obtain the desired representation, we need to make a stronger assumption on the option's payoff $H$.

\begin{ass}\label{ass:H.C3}
$H'$ is globally Lipschitz-continuous.
\end{ass}

Note that the above assumption implies that $\pa_{ss}P$ is absolutely bounded on $[0,T]\times\RR$ (in addition to the properties implied by Assumption \ref{assume:P}).

\begin{proposition}\label{prop:us.rep}
Let us fix an arbitrary initial condition $(s_0,\pi_0,Q)\in\RR^3$ at time $t_0\in[0,T]$ and drop the superscript $(t_0,s_0,\pi_0,Q)$.
Then, under Assumptions \ref{assume:P} and \ref{ass:H.C3}, the following representation holds for $\pa_s u_t:=\pa_s u(t,S_t,\pi^*_t)$:
\begin{align}
&\pa_s u_t = Q\sigma^2 \gamma^2\, \EE_t^{\QQ^{t,s,\pi}}\left.\left[\int_t^Te^{\int_r^t Q\sigma^2 \gamma\pa_{ss}P_vdv }\pa_{ss}P_r(\pi^*_r+Q\pa_s P_r)dr\right]\right|_{(s,\pi)=\left(S_t,\pi^*_t\right)},
\quad t\in[t_0,T].\label{eq:rep2us}
\end{align}
where $\QQ^{t,s,\pi}$ is the probability defined in \eqref{def:Qe}.
\end{proposition}
\begin{proof}
First, recalling \eqref{def:Qe} and noticing that for all $\e\leq \e_0$, with some deterministic $\e_0>0$, we have
$$
\pi_t^{*,0,\e} = \pi_t^{*,0,0},\quad
\ZZ^{0,\e}\left(t,S_t,\pi_t^{*,0,\e}\right)
= \ZZ^{0,0}\left(t,S_t,\pi_t^{*,0,0}\right),
$$
for all $t\in[t_0,T]$, we conclude that Remark \ref{rem:Q.consist} applies for $\e=0$.
Then, using \eqref{eq:sder}, \eqref{eq.Q0T}, \eqref{eq:constproba}, and the fact that $W$ has a drift under $\QQ^{t_0,s_0,\pi_0}$, we obtain:
\begin{align*}
\pa_s u_t&= -\sigma\gamma Q\,\EE^{\QQ^{t,s,\pi}}_t\left.\left[\int_t^T \pa_{ss}P_r dW_r\right]\right|_{(s,\pi)=\left(S_t,\pi^*_t\right)}
= -\sigma\gamma Q\,\EE^{\QQ^{t_0,s_0,\pi_0}}_t\left[\int_t^T \pa_{ss}P_r dW_r\right]\\
&=-\sigma^2\gamma Q\,\EE^{\QQ^{t_0,s_0,\pi_0}}_t\left[\int_t^T \pa_{ss}P_r(\pa_s u_r-\gamma(\pi^*_r+Q\pa_s P_r))dr\right].
\end{align*}
Thus, $\pa_s u$ satisfies 
$$
d(\pa_s u_t)=\sigma^2\gamma Q \pa_{ss}P_t (\pa_s  u_t-\gamma(\pi^*_t+Q\pa_s P_t))dt + d\tilde M_t,
\quad \pa_s u_T=0,
$$
where $\tilde M$ is a $\QQ^{t_0,s_0,\pi_0}$ martingale. We solve this BSDE for $\pa_s u$ and apply Remark \ref{rem:Q.consist} once more to obtain \eqref{eq:rep2us}.
\qed
\end{proof}

\medskip

The following Lemma constitutes the main technical result for computing the desired price expansion.
Its proof is postponed to Appendix B.

\begin{lemma}\label{lem:limitfunctional}
Let $\a$ and $\beta$ be adapted continuous and bounded processes (independent of $\eta$). Denote by $u$ and $\pi^*$, respectively, the log-value function \eqref{eq:defue} and the associated optimal strategy, for an arbitrary (fixed) initial condition $(s,\pi,Q)\in\RR^3$ at time $t=0$. Define $\Gamma$ by 
$$
d\Gamma_t=\a_t(\pa_s u_t-\gamma(\pi_t^*+Q\pa_s P_t))dt+\beta_t d\tilde W_t,
$$
with arbitrary (fixed) $\Gamma_0\in\RR$, and with $\tilde W_t$ being a $\QQ:=\QQ^{0,s,\pi}$-Brownian motion. 
Then, under Assumptions \ref{assume:P} and \ref{ass:H.C3}, as $\eta\rightarrow0$,
\begin{align}
&\EE^\QQ\left[\int_0^T\Gamma_t(\pi_t^*+Q\pa_s P_t)dt\right]=\frac{1}{\kappa}\Gamma_0(\pi+Q\pa_s P(0,s))+\int_0^T\frac{1}{\kappa} \EE\left[Q\sigma\pa_{ss}P_r\beta_r \right]dr\notag\\
&\qquad\qquad+ Q\gamma\sigma^2\int_0^T\EE^\QQ\left[(\pi_t^*+Q\pa_s P_t)\pa_{ss}P_t\int_0^t \Gamma_r e^{\int_t^r Q\sigma^2 \gamma \pa_{ss}P_v dv}dr\right]dt+o(\sqrt{\eta}).\label{eq.expansion.DGamma.mainRep}
\end{align}
\end{lemma}
\begin{remark}
Although it is omitted in the above notation, $\QQ$ also depends on $\eta$. In particular, in the second line of \eqref{eq.expansion.DGamma.mainRep}, the quantities $\QQ$, $(\pi_t^*+Q\pa_s P_t)$, and $\Gamma_r$, all depend on $\eta$.  
\end{remark}

\smallskip

Lemma \ref{lem:limitfunctional} is the main tool for the small impact asymptotic expansion derived in this section. It describes the behavior of the functional 
$$\Gamma\mapsto \EE^\QQ\left[\int_0^T (\pi_t^*+Q\pa_s P_t) \Gamma_t dt\right] $$ in the small $\eta$, or large $\kappa$, regime. 
Note that, in this regime, the function $m$ is large and, thanks to \eqref{eq:duffisiondeviation}, the process $(\pi^*+Q\pa_s P)$, which is the optimally controlled deviation from the frictionless hedge $Q\pa_s P_t$, is strongly mean reverting around zero. The process $(\pi_t^*+Q\pa_s P_t)/\eta^{1/4}$ is, in fact, the so called fast variable mentioned in \cite{BCE,MMS,ST}. However, unlike the latter papers, herein we do not use the viscosity solution methods to characterize the limiting behavior of $(\pi_t^*+Q\pa_s P_t)$. This is due to the fact that, in this work, we establish an expansion for the marginal utility indifference price $p^*$, as opposed to the value function, and the PDE describing the derivatives of $u$ lacks the crucial non-degeneracy property in the state variable $\pi$.  
Thus, herein, we develop a novel methodology for deriving the desired expansion, which relies on the direct probabilistic analysis of the associated optimal control problem, carried out in Section \ref{se:hedging}, and, in particular, on the representation \eqref{def:R} in Proposition \ref{lem:boundR}.

\smallskip

\begin{theorem}\label{theo:expansion}
Let Assumptions \ref{assume:P}, \ref{ass:H.C3} and convention \eqref{rescalingliquidation} hold. Then, the marginal utility indifference price $p^*$ has the following representation for all $(t,s,\pi,Q)\in [0,T]\times \RR^3$,  as $\eta \to 0$:
\begin{align}\label{eq:expansionprice}
p^*(t,s,\pi,Q)
=&P(t,s)-{Q\sqrt{2\eta \gamma\sigma^2}}\int_t^T\EE_{t,s}\left[  \sigma^2(\pa_{ss}P_r)^2 \right]dr\notag\\
&-\sqrt{2\eta \gamma\sigma^2}(\pi+Q\pa_s P(t,s)){\pa_sP(t,s)}+o(\sqrt{\eta}).
\end{align}
\end{theorem}
\begin{proof}
Without loss of generality we prove the expansion at $t=0$. 
We fix $(S_0,\pi_0,Q_0)$ and drop these superscripts. Due to \eqref{eq.expansion.pstar.rep1}, we have 
\begin{align*}
p^*(t,s_0,\pi_0,Q_0)&=P(0,S_0)+\sigma \EE\left[e^{\int_0^T \ZZ\left(r,S_r,\pi^*_r \right)dW_r
- \frac{1}{2}\int_t^T (\ZZ)^2\left(r,S_r,\pi_r^*\right) dr}\int_0^T \pa_s P_rdW_r\right]\\
&= P(0,S_0)+\sigma^2 \EE^\QQ\left[ \int_0^T  \pa_s P_r(\pa_s u(r,S_r,\pi^*_r)-\gamma(\pi^*_r+Q_0\pa_s P_r))dr\right]\label{eq:expexp}\\
&= P(0,S_0)-\sigma^2 \gamma\EE^\QQ\left[\int_0^T\tilde \Gamma_{r}(\pi^*_r+Q_0\pa_s P_r)dr\right],
\end{align*}
where 
$$
\tilde \Gamma_{r}:=\pa_sP_r-Q_0\sigma^2 \gamma\pa_{ss}P_r\int_0^r \pa_sP_h e^{\int_r^h Q_0\sigma^2 \gamma\pa_{ss}P_vdv }dh,
$$
and we have used \eqref{eq:rep2us} to obtain the last equality.

Recall that $\pa_sP_r$ follows
$$d(\pa_sP_t)=\sigma^2\pa_{ss}P_t(\pa_s u_t-\gamma(\pi_t^*+Q\pa_s P_t))dt+\sigma \pa_{ss}P_td\tilde W_t,$$
with a $\QQ$-Brownian motion $\tilde W$.
Applying Lemma \ref{lem:limitfunctional} to $\Gamma_t:=\pa_sP_t$, 
we obtain
\begin{align}
&\EE^\QQ\left[\int_0^T\pa_sP_t(\pi_t^*+Q_0\pa_s P_t)dt\right]=\frac{1}{\kappa}\pa_sP_0(\pi_0+Q_0\pa_s P_0)+\int_0^T\frac{1}{\kappa} \EE\left[Q_0(\sigma\pa_{ss}P_r)^2 \right]dr\notag\\
&\qquad\qquad+ Q_0\gamma\sigma^2\int_0^T\EE^\QQ\left[(\pi_t^*+Q_0\pa_s P_t)\pa_{ss}P_t\int_0^t \pa_sP_r e^{\int_t^r Q\sigma^2 \gamma \pa_{ss}P_v dv}dr\right]dt+o(\kappa^{-1}).
\end{align}
Therefore,
\begin{align*}
&\EE^\QQ\left[\int_0^T\tilde \Gamma_{t}(\pi^*_t+Q_0\pa_s P(t,S_t))dt\right]\\
&=\EE^\QQ\left[\int_0^T\left(\pa_s P_t-Q_0\gamma\sigma^2\pa_{ss}P_t\int_0^t \pa_sP_r e^{\int_t^r Q\sigma^2 \gamma \pa_{ss}P_v dv}dr\right)(\pi_t^*+Q_0\pa_s P_t)dt\right]\\
&=\frac{1}{\kappa}\pa_sP(0,S_0)(\pi_0+Q_0\pa_s P(0,S_0))+\int_0^T\frac{1}{\kappa} \EE\left[Q_0(\sigma\pa_{ss}P_r)^2 \right]dr+o(\kappa^{-1}).
\end{align*}
Collecting the above and recalling \eqref{eq.kappa.m.def} we complete the proof.
\qed
\end{proof}

\medskip

The asymptotic expansion of the marginal utility indifference price, given by the right hand side of \eqref{eq:expansionprice}, has three components.
\begin{enumerate}
\item[i)] The frictionless, or fundamental, price $P(t,s)$.
\item[ii)] A term of order $\sqrt{\eta}$, proportional to the expected cumulative (frictionless) Gamma of the option.
\item[iii)] Another term of order $\sqrt{\eta}$, which is proportional to the (frictionless) Delta of the option multiplied by the deviation of the current position from the optimal frictionless one, $(\pi+Q\pa_s P(t,s))$.
\end{enumerate}

It is important to note that, along the optimal inventory path $\pi^*$, the deviation $(\pi^*_t+Q\pa_s P(t,S_t))$ in fact converges to zero as $\eta\rightarrow0$. Hence, if the agent acts optimally, the last term in the expansion \eqref{eq:expansionprice} for $p^*(r,S_r,\pi^*_r,Q)$ becomes negligible compared to the second one. (This term is only relevant for the cases where thel inventory level $\pi$ is far from the target frictionless value: e.g., at the initial moment when the agent starts hedging.) Thus, along the optimal trajectory $\pi^*$, we expect
\begin{align}\label{eq.DavisPrice.expand.rev.1}
p^*(t,s,\pi,Q)
\approx&P(t,s)-{Q\sqrt{2\eta \gamma\sigma^2}}\int_t^T\EE_{t,s}\left[  \sigma^2(\pa_{ss}P_r)^2 \right]dr,
\end{align}
as $\eta\rightarrow0$.
As the right hand side of the above is an affine function of the option's inventory $Q$, the above representation implies that, in the leading order, the price impact in the option's market is linear and permanent, with the impact coefficient at time $t$ being
\begin{equation}\label{eq.expansion.impact.deriv}
{\sqrt{2\eta \gamma\sigma^2}}\int_t^T\EE_t\left[  \sigma^2(\pa_{ss}P_r)^2\right]dr \geq 0.
\end{equation}

The representation \eqref{eq.DavisPrice.expand.rev.1} also reveals that the marginal indifference price (for small $\eta$) is decreasing in the agent's inventory $Q$. In particular, $p^*$ is expected to be below the frictionless price $P$ if the agent is long the option (i.e., $Q>0$), and to be above the frictionless price if the agent is short the option (i.e., $Q<0$). One explanation of such relationship between $p^*$ and $Q$ is that, in a price impact model, the hedging cost of $Q$ shares of the option is expected to be convex in $Q$. Indeed, the cost of hedging $\Delta Q>0$ shares of the option (caused by the price impact) will increase if the agent also needs to hedge the additional $Q>0$ shares of the option, as the latter hedging will push the underlying price further in the same direction as the hedging of the $\Delta Q$ shares, at every trade. As a result, the higher is the agent's inventory level $Q$, the less she is willing to buy the option and the more she is willing to sell it.

\subsection{Appendix A}

Denote by $\mathcal{T}_{[t,s]}$ the set of all $\FF^t$-stopping times with values in $[t,s]$.

\smallskip

{\bf\cite[Assumption A]{BouchardTouzi}}. For all $(t,s,\pi,x)\in[0,T]\times\RR^3$ and all $\nu \in \cA^\e(t,T)$, the following holds.
\begin{itemize}
\item[A1] (independence). The processes $(S^{t,s},\pi^{\nu,t,\pi},X^{\nu,t,s,x})$ are $\FF^t$-progressively measurable.
\item[A2] (causality). For any $\tilde\nu \in \cA^\e(t,T)$, $\tau\in\mathcal{T}_{[t,T]}$, and $A\in\mathcal{F}^t_\tau$, we have: if $\nu=\tilde\nu$ on $[t,\tau]$ and $\nu\bone_A=\tilde\nu\bone_A$ on $(\tau,T]$, then 
$$
(S^{t,s},\pi^{\nu,t,\pi},X^{\nu,t,s,x})\bone_A=(S^{t,s},\pi^{\tilde\nu,t,\pi},X^{\tilde\nu,t,s,x})\bone_A.
$$
\item[A3] (stability under concatenation). For every $\tilde\nu \in \cA^\e(t,T)$ and $\theta\in\mathcal{T}_{[t,T]}$, we have:
$$
\nu\bone_{[t,\theta]} + \tilde\nu\bone_{(\theta,T]} \in \cA^\e(t,T).
$$
\item[A4] (consistency with deterministic initial data). For every $\theta\in\mathcal{T}_{[t,T]}$, we have:
\begin{itemize}
\item[a.] For $\PP$-a.e. $\omega\in\Omega$, there exists $\tilde\nu_\omega\in\cA^\e(\theta(\omega),T)$, s.t.
\begin{align*}
&\EE\left[\left.-e^{-\gamma \left(X^{\nu,t,s,x}_T+\pi^{\nu,t,\pi}_T S^{t,s}_T - l(\pi^{\nu,t,\pi}_T)^2/2 + Q H(S^{t,s}_T)\right)}
\,\right|\,\mathcal{F}^t_{\theta}\right]\\
&\leq -e^{-\gamma(X^{\nu,t,s,x}_\theta+\pi^{\nu,t,\pi}_\theta S^{t,s}_\theta+QP(\theta,S^{t,s}_\theta))}J^\d\left(\theta,S^{t,s}_\theta,\pi^{\nu,t,\pi}_\theta,Q;\tilde\nu_\omega\right).
\end{align*}
\item[b.] For any $s\in[t,T]$, $\theta\in\mathcal{T}_{[t,s]}$, and $\tilde\nu\in \cA^\e(s,T)$, denoting $\bar \nu := \nu\bone_{[t,\theta]} + \tilde\nu\bone_{(\theta,T]}$, we have
\begin{align*}
&\EE\left[\left.-e^{-\gamma \left(X^{\bar\nu,t,s,x}_T+\pi^{\bar\nu,t,\pi}_T S^{t,s}_T - l(\pi^{\bar\nu,t,\pi}_T)^2/2 + Q H(S^{t,s}_T)\right)}
\,\right|\,\mathcal{F}^t_{\theta}\right]\\
&= -e^{-\gamma(X^{\nu,t,s,x}_\theta+\pi^{\nu,t,\pi}_\theta S^{t,s}_\theta+QP(\theta,S^{t,s}_\theta))} J^\d\left(\theta,S^{t,s}_\theta,\pi^{\nu,t,\pi}_\theta,Q;\tilde\nu\right),
\end{align*}
for $\PP$-a.e. $\omega\in\Omega$.
\end{itemize}
\end{itemize}

\section{Appendix B}

{\bf Proof of Lemma \ref{lem:limitfunctional}}.
By direct computation we have 
$$
e^{-\int_r^tm(v)dv}=\frac{\cosh\left(\kappa(T-t)+\phi\right)}{\cosh\left(\kappa(T-r)+\phi\right)}.
$$ 
Next, we denote 
\begin{align*}
\D_t&:=\pi_t^*+{Q}\pa_s P_t,\\
\phi&:=\frac{1}{2}\ln \left(\frac{\bar l\sqrt{ \eta}+\sqrt{2\gamma \sigma^2}}{-\bar l\sqrt{\eta}+\sqrt{2\gamma \sigma^2}}\right)=\bar l\sqrt{\frac{\eta}{2\gamma\sigma^2}}+o(\eta^{1/2}),\\
\frac{A_{r,s}}{\kappa}&:=\int_r^s e^{-\int_r^tm(v)dv}dt=\frac{\tanh\left(\kappa(T-r)+\phi\right)}{\kappa} - \frac{\sinh\left(\kappa(T-s)+\phi\right)}{\kappa\cosh\left(\kappa(T-r)+\phi\right)},\\
\phi_t&:=m(t)-\kappa A_{t,T}=\kappa\frac{\sinh\left(\phi\right)}{\cosh\left( \kappa(T-t)+\phi\right)}\leq \kappa \bar l\sqrt{\frac{\eta}{2\gamma\sigma^2}}\leq \frac{\bar l}{2},\\
\frac{\tilde A_{r,s}}{\kappa}&:=\int_r^s e^{-\int_t^sm(v)dv}dt.
\end{align*}
Noice that for all $0\leq r\leq s\leq T$, $0\leq A_{r,s}\leq 1$, $0\leq \tilde A_{r,s}\leq 1$, and for all $0\leq r < s\leq T$, we have: $A_{r,s}\to 1$ as $\eta \to 0$. 
Then, the feedback representation of $\pi^*$ in Theorem \ref{theod0}, representation \eqref{eq.R.def}, Proposition \ref{lem:boundR}, representation \eqref{def:ZZZ}, and the fact that $\pa_sP_t$ is a martingale under $\PP$, yield 
\begin{align*}
d\D_t=&-m(t)\D_tdt+\frac{\sigma^2}{2\eta}\EE^{\QQ}_t\left[\int_t^Te^{-\int_t^rm(v)dv}\pa_{s} u_r dr\right]dt\notag\\
&-{Q}\kappa^2\EE^{\QQ}_t\left[\int_t^Te^{-\int_t^rm(v)dv}(\pa_s P_r-\pa_s P_t) dr\right]dt+{Q}\pa_sP_t( m(t)-\kappa A_{t,T})dt\notag\\
&+{Q}\sigma^2\pa_{ss}P_t(\pa_s u_t-\gamma \D_t)dt+\sigma{Q}\pa_{ss}P_t d\tilde W_t.
\end{align*}
Using the dynamics of $\pa_s P_t$ and the definition of $\phi_t$, we transform the above into
\begin{align}\label{eq:duffisiondeviation}
d\D_t=&-m(t)\D_tdt+\frac{\sigma^2}{2\eta}\EE^{\QQ}_t\left[\int_t^Te^{-\int_t^rm(v)dv}\pa_{s} u_r dr\right]dt\notag\\
&-{Q}\kappa^2\sigma^2\EE^{\QQ}_t\left[\int_t^Te^{-\int_t^rm(v)dv}\int_t^r\pa_{ss}P_h(\pa_s u_h - \gamma\D_h)dh dr\right]dt+{Q}\pa_sP_t \phi_tdt\notag\\
&+{Q}\sigma^2\pa_{ss}P_t(\pa_s u_t-\gamma \D_t)dt+\sigma{Q}\pa_{ss}P_t d\tilde W_t.
\end{align}
Therefore, 
\begin{align}
d(\Gamma_t\D_t)=&-m(t)\Gamma_t \D_tdt+\frac{\sigma^2}{2\eta}\Gamma_t\,\EE^{\QQ}_t\left[\int_t^Te^{-\int_t^rm(v)dv}\pa_s u_rdr\right]dt\nonumber\\
&-{Q}\kappa^2\Gamma_t\,\EE^{\QQ}_t\left[\int_t^Te^{-\int_t^rm(v)dv}\left(\pa_sP_r-\pa_sP_t\right)dr\right]dt+\Gamma_t{Q}\pa_sP_t \phi_tdt\nonumber\\
&+{Q}\sigma^2\Gamma_t\pa_{ss}P_t(\pa_s u_t-\gamma \D_t)dt+{Q}\Gamma_t\pa_{ss}P_t d\tilde W_t\label{eq.Expansion.GammaD.decomp}\\
&+\a_t(\pa_s u_t-\gamma\D_t)\D_tdt+\D_t\beta_t d\tilde W_t+{Q}\sigma\pa_{ss}P_t\beta_t dt.\nonumber
\end{align}
Due to the boundedness assumption on $\a,\beta$, $P$, and on the partial derivatives of $P$, as well as the boundedness of the optimal control $\nu^*$, the local martingales in the decomposition \eqref{eq.Expansion.GammaD.decomp} of $\Gamma_t\D_t=\Gamma_t(\pi_t^*+{Q}\pa_s P_t)$ are martingales. This decomposition also shows that $\Gamma_t\D_t$ solves a random linear ODE (to derive this ODE, treat the first term in the right hand side of \eqref{eq.Expansion.GammaD.decomp} as a linear function of $\Gamma_t\D_t$ and the rest as exogenously given source term), which we solve explicitly and integrate the solution over $[0,T]$ to obtain:
\begin{equation}
\EE^\QQ\left[\int_0^T\Gamma_t\D_tdt\right]=\frac{A_{0,T}}{\kappa}\Gamma_0\D_0+\int_0^T\EE^\QQ\left[\frac{A_{r,T}}{\kappa} {Q}\Gamma_r\pa_sP_r \phi_r\right] dr\label{eq.expansion.GammaD.rep.2}
\end{equation}
$$
\qquad\qquad+\int_0^TA_{r,T}\frac{\sigma^2}{2\kappa \eta}\EE^\QQ\left[\Gamma_r\int_r^T e^{-\int_r^h m(v)dv}\pa_s u_h dh\right]dr
$$
$$
\qquad\qquad-{Q}\kappa\int_0^T\EE^\QQ\left[A_{r,T}\left(\Gamma_r\int_r^T e^{-\int_r^h m(v)dv}\left(\pa_sP_h-\pa_sP_r\right)dh \right)\right]dr
$$
$$
\qquad\qquad+{Q}\sigma^2\int_0^T\EE^\QQ\left[\int_0^te^{-\int_r^t m(v) dv}\Gamma_r\pa_{ss}P_r(\pa_s u_r-\gamma \D_r)dr\right]dt
$$
$$
\qquad\qquad+\int_0^T\EE^\QQ\left[\int_0^te^{-\int_r^t m(v) dv}\a_r(\pa_s  u_r-\gamma\D_r)\D_rdr\right]dt
$$
$$
\qquad\qquad+\int_0^T\EE^\QQ\left[\int_0^te^{-\int_r^t m(v) dv}{Q}\sigma\pa_{ss}P_r\beta_r dr\right]dt
$$
$$
=\frac{A_{0,T}}{\kappa}\Gamma_0\D_0+\int_0^T\EE^\QQ\left[\frac{A_{r,T}}{\kappa} {Q}\Gamma_r\pa_sP_r \phi_r\right] dr
$$
\begin{equation}
\qquad\qquad+\int_0^T\frac{A_{r,T}}{\kappa} \EE^\QQ\left[{Q}\sigma\pa_{ss}P_r\beta_r \right]dr+\int_0^T \frac{\sigma^2}{2\kappa \eta}\EE^\QQ\left[\pa_s u_h\int_0^h e^{-\int_r^h m(v)dv}{A_{r,T}}\Gamma_rdr\right]dh\label{eq.expansion.GammaD.rep}
\end{equation}
$$
\qquad\qquad-{Q}\sigma^2 \int_0^TA_{r,T}\EE^\QQ\left[\pa_{ss}P_r(\pa_s  u_r-\gamma\D_r)\int_0^re^{-\int_t^r m(v) dv} \Gamma_tA_{t,T} dt\right]dr
$$
$$
\qquad\qquad+{Q}\sigma^2\int_0^T\frac{A_{r,T}}{\kappa} \EE^\QQ\left[\Gamma_r\pa_{ss}P_r(\pa_s  u_r-\gamma \D_r)\right]dr+\int_0^T\frac{A_{r,T}}{\kappa} \EE^\QQ\left[ \a_r(\pa_s  u_r-\gamma\D_r)\D_r\right]dr,
$$
where we (as before) used the fact that $d\pa_s P_t=\sigma \pa_{ss}P_t dW_t$ to represent the term $\pa_sP_s-\pa_sP_r$.

We denote 
\begin{align*}
\bar\Gamma_t&:=\frac{\kappa}{\gamma}\int_0^t e^{-\int_r^t m(v) dv} \Gamma_rA_{r,T}dr,\\
\tilde \Gamma_t&:=-{Q}\sigma^2\pa_{ss}P_tA_{t,T} \int_0^t e^{-\int_r^t m(v) dv} \Gamma_rA_{r,T}dr + \frac{A_{t,T}(\a_t \Delta_t+{Q}\sigma^2\Gamma_t\pa_{ss}P_t)}{\kappa},
\end{align*} 
and group the terms in the right hand side of \eqref{eq.expansion.GammaD.rep} as follows:
\begin{align*}
&\EE^\QQ\left[\int_0^T\Gamma_t\D_tdt\right]=\frac{A_{0,T}}{\kappa}\Gamma_0\D_0+\int_0^T\EE^\QQ\left[\frac{A_{r,T}}{\kappa}{Q}\Gamma_r\pa_sP_r \phi_r\right] dr\\
&+\int_0^T\frac{A_{r,T}}{\kappa} \EE^\QQ\left[{Q}\sigma\pa_{ss}P_r\beta_r \right]dr\\
&+ {Q}\gamma\sigma^2\int_0^T\EE^\QQ\left[\D_r\pa_{ss}P_rA_{r,T}\left(\int_0^re^{-\int_t^r m(v) dv} \Gamma_tA_{t,T} dt-\frac{\Gamma_r}{\kappa} \right)\right]dr\\
&-\frac{\gamma}{\kappa}\int_0^T \EE^\QQ\left[A_{r,T} \a_r\D_r^2\right]dr
+\int_0^T\EE^\QQ\left[(\bar \Gamma_t+\tilde \Gamma_t)\pa_s u_t\right] dt
\end{align*}
Due to \eqref{eq:rep2us}, the last term in the right hand side of the above becomes
$$Q\sigma^2 \gamma^2\EE^\QQ\left[\int_0^T\pa_{ss}P_r\D_r\int_0^r e^{\int_r^t Q\sigma^2 \gamma\pa_{ss}P_vdv }(\bar\Gamma_t+\tilde\Gamma_t)dtdr\right].$$
We now denote $\hat\Gamma_t:=\gamma\int_0^t \bar\Gamma_r e^{\int_t^r {Q}\sigma^2\gamma\pa_{ss}P_v dv}dr$ and $\check\Gamma_t:=\gamma\int_0^t \tilde\Gamma_r e^{\int_t^r {Q}\sigma^2\gamma\pa_{ss}P_v dv}dr$, so that
$$
\EE^\QQ\left[\int_0^T\Gamma_t\D_tdt\right]=\frac{A_{0,T}}{\kappa}\Gamma_0\D_0+\int_0^T\EE^\QQ\left[\frac{A_{r,T}}{\kappa} {Q}\Gamma_r\pa_sP_r \phi_r\right] dr
$$
$$
+\int_0^T\frac{A_{r,T}}{\kappa} \EE^\QQ\left[{Q}\sigma\pa_{ss}P_r\beta_r \right]dr+Q\gamma\sigma^2\int_0^T\EE^\QQ\left[\D_r\pa_{ss}P_r\hat\Gamma_r\right]dr
$$
$$
+ {Q}\gamma\sigma^2\int_0^T\EE^\QQ\left[\D_r\pa_{ss}P_r\left(\check \Gamma_r+A_{r,T}\int_0^re^{-\int_t^r m(v) dv} \Gamma_tA_{t,T} dt-\frac{A_{r,T}\Gamma_r}{\kappa} \right)\right]dr
$$
$$
-\frac{\gamma}{\kappa}\int_0^T \EE^\QQ\left[A_{r,T} \a_r\D_r^2\right]dr.
$$
Lemma \ref{lem:cv4} (stated further in Appendix B) easily yields 
\begin{align}\label{eq:cvsquare}
\int_0^T \EE^\QQ\left[\D_r^4\right]dr+\int_0^T \EE^\QQ\left[\D_r^2\right]dr=o(1).
\end{align}
We now use this result to show that 
\begin{align*}
I_1&:=\int_0^T\EE^\QQ\left[{A_{r,T}}{Q}\Gamma_r\pa_sP_r \phi_rdr\right]=o(1),\\
I_2&:=\int_0^T{A_{r,T}} \EE^\QQ\left[{Q}\sigma\pa_{ss}P_r\beta_r \right]dr=\int_0^T\EE\left[{Q}\sigma\pa_{ss}P_r\beta_r \right]dr+o(1),\\
I_3&:=\int_0^T\EE^\QQ\left[\D_r\pa_{ss}P_r\hat\Gamma_r\right]dr\\
&=\EE^\QQ\left[\int_0^T\pa_{ss}P_r\D_r\int_0^r e^{\int_r^h Q\sigma^2 \gamma\pa_{ss}P_vdv } \Gamma_h dh dr\right]+o(\kappa^{-1}),\\
I_4&:=\int_0^T\EE^\QQ\left[\D_r\pa_{ss}P_r\left(\check\Gamma_r+A_{r,T}\int_0^re^{-\int_t^r m(v) dv} \Gamma_tA_{t,T} dt-\frac{A_{r,T}\Gamma_r}{\kappa} \right)\right]dr=o(\kappa^{-1}),\\
I_5&:=\int_0^T \EE^\QQ\left[A_{r,T} \a_r\D_r^2\right]dr=o(1).
\end{align*}
We treat each term separately. Recall that $0 \leq A_{r,T}\leq 1$ and $\a$ and $\beta$ are uniformly bounded. 
Note also that  $\int_0^T \phi_r^2dr=o(1)$. Direct estimates yield for $I_1$ amd $I_5$ to
\begin{align*}
|I_1|&\leq C \EE^\QQ\left[\int_0^T\Gamma_r^2dr\right]^{1/2}\left(\int_0^T \phi_r^2dr\right)^{1/2},\\
|I_5|&\leq C\EE^\QQ\left[ \int_0^T\D_r^2dr\right].
\end{align*}
We also estimate $\check \Gamma$ so follows 
\begin{align*}
&\left|\check \Gamma_r+A_{r,T}\int_0^re^{-\int_t^r m(v) dv} \Gamma_tA_{t,T} dt-\frac{A_{r,T}\Gamma_r}{\kappa} \right|\\
&\leq  C\left(\sup_{t}(|\tilde \Gamma_t|)+\sup_{t}(1+| \Gamma_t|)\left(\int_0^re^{-\int_t^r m(v) dv} dt+\frac{1+|\Delta_t|}{\kappa}\right)\right)\\
&\leq C\sup_{t}(1+| \Gamma_t|)\left(\sup_{s\in[0,T]}\int_0^se^{-\int_t^s m(v) dv} dt+\int_0^re^{-\int_t^r m(v) dv} dt+\frac{1+|\Delta_t|}{\kappa}\right)\\
&\leq C\frac{\sup_{t}(1+| \Gamma_t|)(1+|\Delta_t|)}{\kappa}
\end{align*}
so that 
$$|I_4|\leq C\EE^\QQ\left[\int_0^T|\D_r|\frac{(1+|\Delta_t|)\sup_{t}(1+| \Gamma_t|)}{\kappa}dr\right]\leq \frac{C}{\kappa}\EE^\QQ\left[\int_0^T|\D_r^2|+|\D_r^4|dr\right]^{1/2}\EE^\QQ\left[{\sup_{t}| \Gamma_t|}^2\right]^{1/2}.$$

Thus, \eqref{eq:cvsquare} and the boundedness of the characteristics of $\Gamma$ imply: 
$$|I_1|+\kappa |I_4|+|I_5|=o(1). $$
Next, we expand $I_3$ and write it as follows:
\begin{align*}
I_3&=\kappa\int_0^T\EE^\QQ\left[\D_r\pa_{ss}P_r\int_0^r  e^{\int_r^h {Q}\sigma^2\gamma\pa_{ss}P_v dv}\int_0^h e^{-\int_t^h m(v) dv} \Gamma_tA_{t,T}dt dh\right]dr\\
&=\kappa\int_0^T\EE^\QQ\left[\D_r\pa_{ss}P_r\int_0^r  e^{\int_r^t {Q}\sigma^2\gamma\pa_{ss}P_v dv}\int_t^r e^{-\int_t^h m(v) dv+\int_t^h {Q}\sigma^2\gamma\pa_{ss}P_v dv}dh \Gamma_tA_{t,T}dt \right]dr\\
&=\EE^\QQ\left[\int_0^T\pa_{ss}P_r\D_r\int_0^r e^{\int_r^t Q\sigma^2 \gamma\pa_{ss}P_vdv } \Gamma_tA_{t,T}dtdr\right]\\
&+\kappa\EE^\QQ\left[\int_0^T\pa_{ss}P_r\D_r\int_0^re^{\int_r^t Q\sigma^2 \gamma\pa_{ss}P_vdv }\Gamma_tA_{t,T}\right.\\
&\qquad\qquad\left.\left(\int_t^r e^{-\int_t^h m(v) dv+\int_t^h Q\sigma^2 \gamma\pa_{ss}P_vdv }dh -\frac{1}{\kappa}\right)dtdr\right].
\end{align*}
Due to the uniform boundedness of $Q\sigma^2 \gamma\pa_{ss}P_v$, there exists $C_{t,s}$ which is uniformly bounded over $s,t\in[0,T]$ and such that 
$$
|e^{\int_t^s Q\sigma^2 \gamma\pa_{ss}P_vdv }-1|=C_{s,t}|t-s|.
$$
Thus, 
\begin{align*}
\left|\int_t^r e^{-\int_t^s m(v) dv} e^{\int_t^s Q\sigma^2 \gamma\pa_{ss}P_vdv }ds -\frac{1}{\kappa}\right|&\leq \int_t^r e^{-\int_t^s m(v) dv} C_{s,t}|t-s|ds +\frac{1-A_{t,r}}{\kappa}\\
&\leq C \int_t^r \frac{\cosh\left(\kappa(T-s)+\phi\right)}{\cosh\left(\kappa(T-t)+\phi\right)}|t-s|ds+\frac{1-A_{t,r}}{\kappa}\\
&\leq C \int_t^r e^{-\kappa(s-t)}|t-s|ds+\frac{1-A_{t,r}}{\kappa}\\
&\leq C \int_0^{\infty} e^{-\kappa u}udu+\frac{1-A_{t,r}}{\kappa}\leq \frac{C}{\kappa^2}+\frac{1-A_{t,r}}{\kappa}.
\end{align*}
We also have the following bound 
$$
\int_0^r 1-A_{s,r}ds=\int_0^r1-{\tanh\left(\kappa(T-s)+\phi\right)}+\frac{\sinh\left(\kappa(T-r)+\phi\right)}{\cosh\left(\kappa(T-s)+\phi\right)}ds
$$
$$
=r +\frac{1}{\kappa}\ln\left(\frac{\cosh\left(\kappa(T-r)+\phi\right)}{\cosh\left(\kappa T+\phi\right)}\right)
$$
$$
+\frac{1}{\kappa}\left(\arctan\left(\sinh\left(\kappa T+\phi\right)\right)-\arctan\left(\sinh\left(\kappa(T-r)+\phi\right)\right)\right){\sinh\left(\kappa(T-r)+\phi\right)}$$
$$
=\frac{1}{\kappa}\ln\left(\frac{e^{\kappa r}\cosh\left(\kappa(T-r)+\phi\right)}{\cosh\left(\kappa T+\phi\right)}\right)
$$
$$
+\frac{1}{\kappa}\left(\arctan\left(\frac{\sinh\left(\kappa T+\phi\right)-\sinh\left(\kappa(T-r)+\phi\right)}{1+\sinh\left(\kappa T+\phi\right)\sinh\left(\kappa(T-r)+\phi\right)}\right)\right){\sinh\left(\kappa(T-r)+\phi\right)}
$$
$$
\leq\frac{1}{\kappa}\ln\left(\frac{1+e^{-\kappa(2T-2r)-2\phi}}{1+e^{-2\kappa T-2\phi}}\right)
$$
$$
+\frac{1}{\kappa}\frac{\sinh\left(\kappa T+\phi\right)\sinh\left(\kappa(T-r)+\phi\right)-\sinh^2\left(\kappa(T-r)+\phi\right)}{1+\sinh\left(\kappa T+\phi\right)\sinh\left(\kappa(T-r)+\phi\right)}
$$
\begin{equation}
\leq\frac{\ln(2)}{\kappa}+\frac{1}{\kappa}.\label{eq.expansion.A.small}
\end{equation}
Given the definition of $I_3$ and the inequality $1-A_{s,r}\geq 1-A_{s,T}\geq 0$, for $s\leq r\leq T$, the above inequalities, along with Cauchy inequality, yield:
\begin{align*}
&\left|I_3-\EE^\QQ\left[\int_0^T\pa_{ss}P_r\D_r\int_0^r e^{\int_r^s Q\sigma^2 \gamma\pa_{ss}P_vdv } \Gamma_sdsdr\right]\right|\\
&\leq C\EE^\QQ\left[\int_0^T\left|\pa_{ss}P_r\D_r\right|\sup_{t} |\Gamma_t|\int_0^r|1-A_{t,r}|dtdr\right]+o(\kappa^{-1})\\
&\leq C\frac{1}{\kappa}\EE^\QQ\left[\int_0^T\D_r^2dr\right]^{1/2}\EE^\QQ\left[ \sup_{t} |\Gamma_t|^2\right]^{1/2}+o(\kappa^{-1})=o(\kappa^{-1})
\end{align*}
where we have used one more time that $\EE^\QQ\left[\int_0^T\D_r^2dr\right]=o(1)$.

To finish the proof of the lemma it now suffices to prove that 
$$I_2=\int_0^T\EE\left[{Q}\sigma\pa_{ss}P_r\beta_r \right]dr+o(1).$$
In view of \eqref{eq.expansion.A.small}, the above is a direct consequence of the convergence
\begin{align}\label{eq:cvprob}
\EE^{\QQ}\left[X \right]\to \EE\left[X \right],
\end{align}
for all absolutely bounded $X$, and the dominated convergence theorem.
Let us prove \eqref{eq:cvprob}. Thanks to martingale representation theorem, there exists a $\PP$-square integrable $h$ such that
\begin{align*}
\EE^{\QQ}\left[X \right]-\EE\left[X \right]&=\EE^{\QQ}\left[\int_0^T h_t dW_t\right]=\sigma\EE^{\QQ}\left[\int_0^T h_t\pa_s u(t,S_t,\pi_t^*)dt\right]-\gamma\sigma\EE^{\QQ}\left[\int_0^T h_t\D_tdt\right]\\
&=Q\sigma^3 \gamma^2\EE^{\QQ}\left[\int_0^T h_t \int_t^Te^{\int_r^t Q\sigma^2 \gamma\pa_{ss}P_vdv }\pa_{ss}P_r\D_rdrdt\right]-\gamma\sigma\EE^{\QQ}\left[\int_0^T h_t\D_tdt\right]\\
&=\sigma \gamma\EE\left[\int_0^T \frac{d\QQ}{d\PP}\D_r\left(Q\sigma^2 \gamma\pa_{ss}P_r \int_0^rh_t e^{\int_r^t Q\sigma^2 \gamma\pa_{ss}P_vdv }dt-h_r\right)dr\right].
\end{align*}
Thus, using the generalized H\"older inequality with $\frac{1}{4}+\frac{1}{4}+\frac{1}{2}=1$, we have
$$\left|\EE^{\QQ}\left[X \right]-\EE\left[X \right]\right|\leq C \left[\EE^{\QQ} \int_0^T \D_r^4dr\right]^{1/4} \left[\EE\int_0^T \left(\frac{d\QQ}{d\PP}\right)^3dr\right]^{1/4} \left[\EE\int_0^T h_t^{2}dt\right]^{1/2}.$$
Thanks to Lemma \ref{lem:cv4}, to finish the proof of this lemma it now suffices to prove that $\EE\left[\int_0^T \left(\frac{d\QQ}{d\PP}\right)^3dr\right]$ is bounded as $\eta\to 0$. 
Thanks to \eqref{eq.Q0T},
$$\left(\frac{d\QQ}{d\PP}\right)^3=\frac{e^{3\Psi^\d(t,\pi, \nu^{*,0,S_0,\pi_0})+3Q\Gamma(0,S_0) }}{U^3(0,S_0,\pi_0)}.$$
Thus, 
$$\EE\left[\left(\frac{d\QQ}{d\PP}\right)^3\right]\leq \frac{\tilde U(0,S_0,\pi_0)}{U^3(0,S_0,\pi_0)}$$
where $\tilde U$ is defined as in \eqref{def:U} but with $\gamma$ replaced by $3\gamma$. This implies the desired convergence
$$\left|\EE^{\QQ}\left[X \right]-\EE\left[X \right]\right|\leq C \EE^{\QQ}\left[\int_0^T \D_r^4dr\right]^{1/4} \to 0\mbox{    as    }\eta \to 0.$$
 \qed

\medskip

\begin{lemma}\label{lem:cv4}
Under the assumptions of Lemma \ref{lem:limitfunctional}, we have:
\begin{align*}
\int_0^T\EE^\QQ\left[\D_r^4\right]dr\leq C\frac{\D_0^4}{\kappa}+\frac{C}{\kappa^{1/2}},
\end{align*}
for small enough $\eta$ (i.e., large enough $\kappa$), with some constant $C>0$ independent of $\eta$.
\end{lemma}
\begin{proof}
Recall that $\D_r=(\pi_r^*+{Q}\pa_s P_r)$. Equation \eqref{eq:duffisiondeviation}, viewed as a linear random ODE for $\D$, implies that, for $0\leq t_0 \leq t_1\leq T$, 
\begin{align}\label{eq:solD}
\D_{t_1}=&\D_{t_0}e^{-\int_{t_0}^{t_1}m_vdv}+\int_{t_0}^{t_1}e^{-\int_t^{t_1}m_vdv}\frac{\sigma^2}{2\eta}\EE^{\QQ}_t\left[\int_t^Te^{-\int_t^rm(v)dv}\pa_s u_r dr\right]dt+\int_{t_0}^{t_1}e^{-\int_t^{t_1}m_vdv}{Q}\pa_sP_t \phi_tdt\notag\\
&-\int_{t_0}^{t_1}e^{-\int_t^{t_1}m_vdv}{Q}\kappa^2\sigma^2\EE^{\QQ}_t\left[\int_t^Te^{-\int_t^rm(v)dv}\int_t^r\pa_{ss}P_h(\pa_s u_h-\gamma\D_h)dh dr\right]dt\notag\\
&+\int_{t_0}^{t_1}e^{-\int_t^{t_1}m_vdv}{Q}\sigma^2\pa_{ss}P_t(\pa_s u_t-\gamma \D_t)dt+\int_{t_0}^{t_1}e^{-\int_t^{t_1}m_vdv}\sigma{Q}\pa_{ss}P_t d\tilde W_t.
\end{align}
Plugging the above into \eqref{eq:rep2us}, we obtain 
$$
\pa_su_{t_0}=\D_{t_0}Q\sigma^2 \gamma^2\EE_{t_0}^{\QQ}\left[\int_{t_0}^Te^{\int_{t_1}^{t_0} Q\sigma^2 \gamma\pa_{ss}P_vdv }\pa_{ss}P_{t_1}e^{-\int_{t_0}^{t_1}m_vdv}dt_1\right]
$$
$$
+Q\sigma^2 \gamma^2\EE_{t_0}^{\QQ}\left[\int_{t_0}^Te^{\int_{t_1}^{t_0} Q\sigma^2 \gamma\pa_{ss}P_vdv }\pa_{ss}P_{t_1}\right.
$$
$$
\left.\qquad\qquad\qquad\int_{t_0}^{t_1}e^{-\int_t^{t_1}m_vdv}\frac{\sigma^2}{2\eta}\EE^{\QQ}_t\left[\int_t^Te^{-\int_t^rm(v)dv}\pa_s u_r dr\right]dtdt_1\right]$$
$$
-Q\sigma^2 \gamma^2\EE_{t_0}^{\QQ}\left[\int_{t_0}^Te^{\int_{t_1}^{t_0} Q\sigma^2 \gamma\pa_{ss}P_vdv }\pa_{ss}P_{t_1}\right.
$$
$$
\left.\qquad\qquad\qquad\int_{t_0}^{t_1}e^{-\int_t^{t_1}m_vdv}{Q}\kappa^2\sigma^2\EE^{\QQ}_t\left[\int_t^Te^{-\int_t^rm(v)dv}\int_t^r\pa_{ss}P_h(\pa_s u_h-\gamma\D_h)dh dr\right]dtdt_1\right]
$$
$$
+Q\sigma^2 \gamma^2\EE_{t_0}^{\QQ}\left[\int_{t_0}^Te^{\int_{t_1}^{t_0} Q\sigma^2 \gamma\pa_{ss}P_vdv }\pa_{ss}P_{t_1}\int_{t_0}^{t_1}e^{-\int_t^{t_1}m_vdv}{Q}\pa_sP_t \phi_tdtdt_1\right]
$$
$$
+Q\sigma^2 \gamma^2\EE_{t_0}^{\QQ}\left[\int_{t_0}^Te^{\int_{t_1}^{t_0} Q\sigma^2 \gamma\pa_{ss}P_vdv }\pa_{ss}P_{t_1}\int_{t_0}^{t_1}e^{-\int_t^{t_1}m_vdv}{Q}\sigma^2\pa_{ss}P(t,S_t)(\pa_s u_t-\gamma \D_t)dtdt_1\right]
$$
$$
+Q\sigma^2 \gamma^2\EE_{t_0}^{\QQ}\left[\int_{t_0}^Te^{\int_{t_1}^{t_0} Q\sigma^2 \gamma\pa_{ss}P_vdv }\pa_{ss}P_{t_1}\int_{t_0}^{t_1}e^{-\int_t^{t_1}m_vdv}\sigma{Q}\pa_{ss}P_t d\tilde W_tdt_1\right].
$$
Recall the estimate 
$$\kappa\int_r^s e^{-\int_r^tm(v)dv}dt+\kappa\int_r^s e^{-\int_t^sm(v)dv}dt=A_{r,s}+\tilde A_{r,s}\leq C\mbox{ for } 0\leq r\leq T.$$
Thus, 
\begin{align*}
|\pa_su_{t_0}|&\leq\frac{C(1+|\D_{t_0}|)}{\kappa}+C\kappa \EE_{t_0}^{\QQ}\left[\int_{t_0}^T\int_{t_0}^re^{-\int_t^rm(v)dv} dt|\pa_s u_r|dr\right]\\
&+C\kappa\EE_{t_0}^{\QQ}\left[\int_{t_0}^T\int_t^Te^{-\int_t^rm(v)dv}\int_t^r(|\pa_s u_h|+\gamma|\D_h|)dh drdt\right]\\
&+\frac{C}{\kappa}\EE_{t_0}^{\QQ}\left[\int_{t_0}^T(|\pa_s u_t|+\gamma |\D_t|)dt\right]+C\EE_{t_0}^{\QQ}\left[\int_{t_0}^T\left|\int_{t_0}^{t_1}e^{-\int_t^{t_1}m_vdv}\pa_{ss}P_t d\tilde W_t\right|dt_1\right]\\
&\leq\frac{C(1+|\D_{t_0}|)}{\kappa}+C \EE_{t_0}^{\QQ}\left[\int_{t_0}^T|\pa_s u_r| dr\right]\\
&+\frac{C}{\kappa}\EE_{t_0}^{\QQ}\left[\int_{t_0}^T\gamma |\D_t|dt\right]+C\EE_{t_0}^{\QQ}\left[\int_{t_0}^T\left|\int_{t_0}^{t_1}e^{-\int_t^{t_1}m_vdv}\pa_{ss}P_t d\tilde W_t\right|dt_1\right].
\end{align*}
Thus, there exists a family of random variables $(M^{t_0})$, continuous in $t_0$, satisfying $\EE^\QQ_{t_0}M^{t_0}=0$, and such that
\begin{align*}
|\pa_su_{t_0}|&\leq\frac{C(1+|\D_{t_0}|)}{\kappa}+C \int_{t_0}^T|\pa_s u_r| dr+\frac{C}{\kappa}\int_{t_0}^T |\D_t|dt+C\int_{t_0}^T\left|\int_{t_0}^{t_1}e^{-\int_t^{t_1}m_vdv}\pa_{ss}P_t d\tilde W_t\right|dt_1+M^{t_0}.
\end{align*}
Applying Gronwall's lemma backwards in time on $[0,T]$, we deduce 
\begin{align*}
|\pa_su_{t_0}|&\leq\left(\frac{C(1+|\D_{t_0}|)}{\kappa}+\frac{C}{\kappa}\int_{t_0}^T |\D_t|dt+C\int_{t_0}^T\left|\int_{t_0}^{t_1}e^{-\int_t^{t_1}m_vdv}\pa_{ss}P_t d\tilde W_t\right|dt_1+M^{t_0}\right)\\
&+C\int_{t_0}^T\left(\frac{C(1+|\D_{r}|)}{\kappa}+\frac{C}{\kappa}\int_{r}^T |\D_t|dt+C\int_{r}^T\left|\int_{r}^{t_1}e^{-\int_t^{t_1}m_vdv}\pa_{ss}P_t d\tilde W_t\right|dt_1+M^{r}\right)e^{C(T-r)}dr.
\end{align*}
Taking expectation, we obtain 
\begin{align*}
|\pa_su_{t_0}|&\leq \frac{C(1+|\D_{t_0}|)}{\kappa}+\frac{C}{\kappa}\EE_{t_0}^{\QQ}\left[\int_{t_0}^T |\D_t|dt\right]+C\int_{t_0}^T\EE_{t_0}^{\QQ}\left|\int_{t_0}^{t_1}e^{-\int_t^{t_1}m_vdv}\pa_{ss}P_t d\tilde W_t\right|dt_1\\
&+C\int_{t_0}^T\int_{r}^T\EE_{t_0}^{\QQ}\left|\int_{r}^{t_1}e^{-\int_t^{t_1}m_vdv}\pa_{ss}P_t d\tilde W_t\right|dt_1dr.
\end{align*}
Note that 
$$\EE_{t_0}^{\QQ}\left|\int_{r}^{t_1}e^{-\int_t^{t_1}m_vdv}\pa_{ss}P_t d\tilde W_t\right|\leq C\left[\int_{r}^{t_1}e^{-2\int_t^{t_1}m_vdv}dt\right]^{1/2}\leq \frac{C}{\sqrt{\kappa}}.$$
We finally obtain 
\begin{align}\label{eq:contpa}
|\pa_su_{t_0}|\leq \frac{C(\sqrt{\kappa}+|\D_{t_0}|)}{\kappa}+\frac{C}{\kappa}\EE_{t_0}^{\QQ}\left[\int_{t_0}^T |\D_t|dt\right].
\end{align}
Injecting \eqref{eq:contpa} into \eqref{eq:solD} yields
\begin{align*}
\int_{t_0}^T\EE^\QQ_{t_0}|\D_{t_1}|dt_1\leq&C\frac{\sqrt{\kappa}+|\D_{{t_0}}|}{\kappa}+C \int_{{t_0}}^{T}\int_t^Te^{-\int_t^rm(v)dv}\EE^{\QQ}_{t_0}|\D_r |drdt\\
&\leq C\frac{\sqrt{\kappa}+|\D_{{t_0}}|}{\kappa}+C \int_{{t_0}}^{T}\EE^{\QQ}_{t_0}\left[|\D_r |\right]\int_{t_0}^re^{-\int_t^rm(v)dv}dtdr\\
&\leq C\frac{\sqrt{\kappa}+|\D_{{t_0}}|}{\kappa}+\frac{C}{\kappa} \int_{{t_0}}^{T}\EE^{\QQ}_{t_0}|\D_r |dr,
\end{align*}
which implies that 
\begin{align*}
\int_{t_0}^T\EE^\QQ_{t_0}|\D_{t_1}|dt_1\leq&C\frac{\sqrt{\kappa}+|\D_{{t_0}}|}{\kappa}.
\end{align*}
Combined with \eqref{eq:contpa} the above inequality yields
\begin{align}\label{eq:contpa2}
|\pa_su_{t_0}|\leq \frac{C(\sqrt{\kappa}+|\D_{t_0}|)}{\kappa}.
\end{align}

Next, we denote $$A^4_{r,s}:=4 \kappa\int_r^s e^{-4\int_r^tm(v)dv}dt.$$ Similarly to $A_{r,s}$, for $0\leq r\leq s$, we have: $A^{4}_{r,s}\leq 1$. 
Applying It\^o's lemma to $\D^4_t$ and using \eqref{eq:duffisiondeviation}, we derive a linear random ODE for $\D^4_t$, similar to \eqref{eq.Expansion.GammaD.decomp}. We solve it and integrate over $[0,T]$ to obtain (similar to the derivation of \eqref{eq.expansion.GammaD.rep.2}):
$$
\int_0^T\EE^\QQ\left[\D_r^4\right]dr=\D_0^4\frac{A^4_{0,T}}{\kappa}+\int_0^T\frac{A^4_{t,T}}{\kappa}\EE^\QQ\left[\frac{\sigma^2\D_t^{3}}{2 \eta}  \int_t^Te^{-\int_t^rm(v)dv}\pa_s u_rdr\right]dt
$$
$$
-\int_0^T{\kappa A^4_{t,T}}\EE^\QQ\left[\D_t^{3}{Q}\sigma^2 \int_t^Te^{-\int_t^h m(v)dv}\int_t^h (\pa_s u_r -\gamma \D_r)drdh\right]dt
$$
$$
+\int_0^T\frac{A^4_{t,T}}{\kappa}\EE^\QQ\left[\D_t^{3}{Q}\pa_sP_t \phi_t \right]dt+\int_0^T\frac{A^4_{t,T}}{\kappa}\EE^\QQ\left[\D_t^{3}{Q}\sigma^2\pa_{ss}P_t(\pa_s u_t-\gamma \D_t\right]dt
$$
$$
+\int_0^T\frac{A^4_{t,T}}{\kappa}\EE^\QQ\left[\frac{3}{2}\D_t^{3}{Q^2}\sigma^2(\pa_{ss}P_t)^2\right]dt.
$$
By Fubini's theorem we have 
\begin{align*}
\int_t^Te^{-\int_t^h m(v)dv}\int_t^h (\pa_s u_r -\gamma \D_r)drdh&=\int_t^Te^{-\int_t^r m(v)dv}(\pa_s u_r -\gamma \D_r)\int_r^Te^{-\int_r^h m(v)dv} dhdr\\
&=\int_t^Te^{-\int_t^r m(v)dv}(\pa_s u_r -\gamma \D_r)\frac{A_{r,T}}{\kappa}dr.
\end{align*}
Collecting the last two equations above, using \eqref{eq:contpa2}, the boundedness of $A_{r,T}$, and the definition of $\kappa$, we obtain:
\begin{align*}
\int_0^T \EE^\QQ\left[\D_r^4\right]dr\leq& C\frac{\D_0^4}{\kappa}+\frac{C}{\sqrt{\kappa}}\int_0^T\EE^\QQ\left[|\D_t|^{3}\right]dt+C\int_0^T\EE^\QQ\left[ |\D_t|^{3}  \int_t^Te^{-\int_t^rm(v)dv}|\D_r|dr\right]dt\\
&+C\int_0^T \EE^\QQ\left[|\D_t|^3\int_t^Te^{-\int_t^rm(v)dv} (|\pa_s u_r| +\gamma |\D_r|)dr\right]dt\\
&+\frac{C}{\kappa}\int_0^T\EE^\QQ\left[|\D_t|^{3}(1+|\pa_s u_t|+\gamma |\D_t|\right]dt.
\end{align*}
Using \eqref{eq:contpa2}, H\"older inequality, and Jensen's inequality, we obtain: 
\begin{align*}
\int_0^T\EE^\QQ\left[\D_r^4\right]dr\leq& C\frac{\D_0^4}{\kappa}+C\left[\EE^\QQ \int_0^T |\D_t|^{4} dt\right]^{3/4}\left(\frac{1}{\sqrt\kappa}+ \left[\EE^\QQ\int_0^T   \int_t^Te^{-4\int_t^rm(v)dv}|\D_r|^4drdt\right]^{1/4}\right)\\
&+\frac{C}{\kappa}\left[\EE^\QQ \int_0^T |\D_t|^{4} dt\right]^{3/4}+\frac{C}{\kappa}\int_0^T\EE^\QQ\left[|\D_t|^{4}\right]dt.
\end{align*}
We combine the last term with the left hand side and apply Fubini's theorem one more time to conclude 
\begin{align*}
\int_0^T\EE^\QQ\left[\D_r^4\right]dr\leq& C\frac{\D_0^4}{\kappa}+C\left[\EE^\QQ \int_0^T |\D_t|^{4} dt\right]^{3/4}\left(\frac{1}{\sqrt\kappa}+ \frac{1}{\kappa^{1/4}}\left[\EE^\QQ\int_0^T   |\D_r|^4dr\right]^{1/4}\right).
\end{align*}
The above estimate, for $\kappa $ large enough, yields the desired inequality.
\qed
\end{proof}

\bibliographystyle{plain}
\bibliography{DerivPriceImpact}

\begin{thebibliography}{10}

\bibitem{AlmgrenChriss}
R.~Almgren and N.~Chriss.
\newblock Optimal execution of portfolio transactions.
\newblock {\em Journal of Risk}, 3(2):5--39, 2001.

\bibitem{BankBaum}
P.~Bank and D.~Baum.
\newblock Hedging and portfolio optimization in financial markets with a large
  trader.
\newblock {\em Journal of Risk}, 14(1):1--18, 2004.

\bibitem{BankVS}
Peter Bank, H~Mete Soner, and Moritz Vo{\ss}.
\newblock Hedging with temporary price impact.
\newblock {\em Mathematics and financial economics}, 11(2):215--239, 2017.

\bibitem{BV}
Peter Bank and Moritz Vo{\ss}.
\newblock Linear quadratic stochastic control problems with stochastic terminal
  constraint.
\newblock {\em SIAM Journal on Control and Optimization}, 56(2):672--699, 2018.

\bibitem{Barles}
G.~Barles.
\newblock A new stability result for viscosity solutions of nonlinear parabolic
  equations with weak convergence in time.
\newblock {\em Comptes Rendus Math\'ematique}, 343(3):173--178, 2006.

\bibitem{BCE}
Erhan Bayraktar, Thomas Caye, and Ibrahim Ekren.
\newblock Asymptotics for small nonlinear price impact: a pde homogenization
  approach to the multidimensional case.
\newblock {\em arXiv preprint arXiv:1811.06650}, 2018.

\bibitem{BSV}
M.~Beiglboeck, W.~Schachermayer, and B.~Veliyev.
\newblock A short proof of the {Doob--Meyer} theorem.
\newblock {\em Stochastic Processes and their Applications}, 122(4):1204--1209,
  2012.

\bibitem{BSh}
Maxim Bichuch and Steven Shreve.
\newblock Utility maximization trading two futures with transaction costs.
\newblock {\em SIAM Journal on Financial Mathematics}, 4(1):26--85, 2013.

\bibitem{Loeper}
B.~Bouchard, G.~Loeper, H.~M. Soner, and C.~Zhou.
\newblock Second order stochastic target problems with generalized market
  impact.
\newblock {\em arXiv preprint arXiv:1806.08533v1}, 2018.

\bibitem{BouchardTouzi}
B.~Bouchard and N.~Touzi.
\newblock Weak dynamic programming principle for viscosity solutions.
\newblock {\em SIAM Journal on Control and Optimization}, 49(3):948--962, 2011.

\bibitem{Carmona}
R.~Carmona.
\newblock {\em Indifference Pricing: Theory and Applications}.
\newblock Princeton University Press, 2009.
\newblock Princeton Series in Financial Engineering. Edited by R. Carmona and
  E. Cinlar.

\bibitem{CHM}
Thomas Cay{\'e}, Martin Herdegen, and Johannes Muhle-Karbe.
\newblock Trading with small nonlinear price impact.
\newblock {\em Available at SSRN}, 2018.

\bibitem{C}
Cid-Araujo and J.~Angel.
\newblock The uniqueness of fixed points for decreasing operators.
\newblock {\em Applied mathematics letters}, 17(7):861, 2004.

\bibitem{CvitanicMa}
J.~Cvitani{\'c} and J.~Ma.
\newblock Hedging options for a large investor and forward-backward {SDEs}.
\newblock {\em Annals of Applied Probability}, 6(2):370--398, 1996.

\bibitem{Davis}
M.~H.~A. Davis.
\newblock Option pricing in incomplete markets.
\newblock In {\em Mathematics of Derivative Securities}, pages 216--227.
  Cambridge University Press, Cambridge, 1997.

\bibitem{delbaen1994general}
Freddy Delbaen.
\newblock A general version of the fundamental theorem of asset pricing.
\newblock {\em Mathematische annalen}, 300(1):463--520, 1994.

\bibitem{ENequil}
I.~Ekren and S.~Nadtochiy.
\newblock Equilibrium price of an option in {Almgren-Chriss} model with
  temporary price impact and competing market makers.
\newblock {\em Forthcoming}, 2020.

\bibitem{EMK}
Ibrahim Ekren and Johannes Muhle-Karbe.
\newblock Portfolio choice with small temporary and transient price impact.
\newblock {\em Mathematical Finance}, 2019.

\bibitem{GP1}
Nicolae G{\^a}rleanu and Lasse~Heje Pedersen.
\newblock Dynamic trading with predictable returns and transaction costs.
\newblock {\em The Journal of Finance}, 68(6):2309--2340, 2013.

\bibitem{GP2}
Nicolae G{\^a}rleanu and Lasse~Heje Pedersen.
\newblock Dynamic portfolio choice with frictions.
\newblock {\em Journal of Economic Theory}, 165:487--516, 2016.

\bibitem{GK}
A.~Greco and B.~Kawohl.
\newblock Log-concavity in some parabolic problems.
\newblock {\em Electronic Journal of Differential Equations}, 1999(19):1--12,
  1999.

\bibitem{GW}
Paolo Guasoni and Marko Weber.
\newblock Nonlinear price impact and portfolio choice.
\newblock {\em Available at SSRN 2613284}, 2015.

\bibitem{GuP}
Olivier Gu{\'e}ant and Jiang Pu.
\newblock Option pricing and hedging with execution costs and market impact.
\newblock {\em Mathematical Finance}, 27(3):803--831, 2017.

\bibitem{KM}
Jan Kallsen and Johannes Muhle-Karbe.
\newblock Option pricing and hedging with small transaction costs.
\newblock {\em Mathematical Finance}, 25(4):702--723, 2015.

\bibitem{Kramkov}
D.~Kramkov and M.~Sirbu.
\newblock Sensitivity analysis of utility-based prices and risk-tolerance
  wealth processes.
\newblock {\em Annals of Applied Probability}, 16(4):2140--2194, 2006.

\bibitem{KrylovNonlin}
N.V. Krylov.
\newblock {\em Nonlinear Elliptic and Parabolic Equations of the Second Order}.
\newblock D. Reidel Publishing Company, 1987.

\bibitem{AlmgrenHedging}
T.~M. Li and R.~Almgren.
\newblock Option hedging with smooth market impact.
\newblock {\em Market Microstructure and Liquidity}, 2(1), 2016.

\bibitem{DL}
F.~Da Lio and O.~Ley.
\newblock Uniqueness results for second-order {Bellman--Isaacs} equations under
  quadratic growth assumptions and applications.
\newblock {\em SIAM journal on control and optimization}, 45(1):74--106, 2006.

\bibitem{Yong}
H.~Liu and J.~M. Yong.
\newblock Option pricing with an illiquid underlying asset market.
\newblock {\em Journal of Economic Dynamics and Control}, 29:2125--2156, 2005.

\bibitem{MS}
P.~Milgrom and I.~Segal.
\newblock Envelope theorems for arbitrary choice sets.
\newblock {\em Econometrica}, 70(2):583--601, 2002.

\bibitem{MMS}
L.~Moreau, J.~Muhle-Karbe, and H.~M. Soner.
\newblock Trading with small price impact.
\newblock {\em Mathematical Finance}, 27(2):350--400, 2017.

\bibitem{PST}
Dylan Possama{\"\i}, H~Mete Soner, and Nizar Touzi.
\newblock Homogenization and asymptotics for small transaction costs: the
  multidimensional case.
\newblock {\em Communications in Partial Differential Equations},
  40(11):2005--2046, 2015.

\bibitem{ST}
H.~M. Soner and N.~Touzi.
\newblock Homogenization and asymptotics for small transaction costs.
\newblock {\em SIAM Journal on Control and Optimization}, 51(4):2893--2921,
  2013.

\end{thebibliography}
\end{document}